\documentclass[graybox]{svmult}

\usepackage{mathptmx}       
\usepackage{helvet}         
\usepackage{courier}        
\usepackage{type1cm}        
\usepackage{makeidx}         
\usepackage{graphicx}        
\usepackage{multicol}        
\usepackage[bottom]{footmisc}

\usepackage{url}
\usepackage{acronym}
\usepackage{amssymb}

\makeindex             

\begin{document}

\title*{Game theoretic modeling of pilot behavior during mid-air encounters}
\author{Ritchie Lee and David Wolpert}
\institute{Ritchie Lee \at Carnegie Mellon University Silicon Valley, NASA Ames Research Park, Mail Stop 23-11, Moffett Field, CA, 94035 \email{ritchie.lee@sv.cmu.edu}
\and David Wolpert \at NASA Ames Research Center, Mail Stop 269-1, Moffett Field, CA, 94035 
\email{david.h.wolpert@nasa.gov}}
%
%
\maketitle


\newacro{tcas}[TCAS]{Traffic Alert and Collision Avoidance System}
\newacro{ra}[RA]{Resolution Advisory}
\newacroplural{ra}[RAs]{Resolution Advisories}
\newacro{hitl}[HITL]{Human-In-The-Loop}
\newacro{nmac}[NMAC]{near mid-air collision}
\newacro{irac}[IRAC]{Integrated Resilient Aircraft Control}


\abstract{We show how to combine Bayes nets and game theory
to predict the behavior of hybrid systems involving both humans and automated components. We call this novel framework ``Semi Network-Form Games," and
illustrate it by predicting aircraft pilot behavior in potential near mid-air collisions. 
At present, at the beginning of such potential collisions,  a collision avoidance system in the aircraft cockpit advises the pilots what to do to avoid the collision.  However studies of mid-air encounters have found wide variability in pilot responses to avoidance system advisories.  In particular, pilots rarely perfectly execute the recommended maneuvers, despite the fact that the collision avoidance system's effectiveness relies on their doing so.
Rather pilots decide their actions based on all information available to them (advisory, instrument readings, visual observations).  We show how to build this aspect into a semi network-form game model of the encounter and then present computational simulations of the resultant model.}

\acresetall
\section{Introduction}
\label{sec:intro}
Bayes nets have been widely investigated and commonly used to describe stochastic systems \protect{\cite{BishopBook,Darwiche09,RussellBook}}.  Powerful techniques already exist for the manipulation, inference, and learning of probabilistic networks.  Furthermore, these methods have been well-established in many domains, including expert systems, robotics, speech recognition, and networking and communications \protect{\cite{KollerBook}}.  On the other hand, game theory is frequently used to describe the behavior of interacting humans \cite{Crawford02,Crawford08}.  A vast amount of experimental literature exists (especially in economic contexts, such as auctions and negotiations), which analyze and refine human behavior models \protect{\cite{CamererBook, Caplin10,Selten08}}.  These two fields have traditionally been regarded as orthogonal bodies of work.  However, in this work we propose to create a modeling framework that leverages the strengths of both.

Building on earlier approaches \protect{\cite{Camerer10,Koller03}}, we introduce a novel framework, ``Semi Network-Form Game," (or ``semi net-form game") that combines Bayes nets and game theory to model hybrid systems.  We use the term ``hybrid systems" to mean such systems that may involve multiple interacting human and automation components.  The semi network-form game is a specialization of the complete framework ``network-form game," formally defined and elaborated in \protect{\cite{WolpertNFG}}.

The issue of aircraft collision avoidance has recently received wide attention from aviation regulators due to some alarming \ac{nmac} statistics \protect{\cite{SfeArticle}}. Many discussions call into question the effectiveness of current systems, especially that of the onboard collision avoidance system.  This system, called ``\ac{tcas}," is associated with many weaknesses that render it increasingly less effective as traffic density grows exponentially.  Some of these weaknesses include complex contorted advisory logic, vertical only advisories, and unrealistic pilot models.  In this work, we demonstrate how the collision avoidance problem can be modeled using a semi net-form game, and show how this framework can be used to perform safety and performance analyses.

The rest of this chapter is organized as follows.  In Section~\ref{sec:snfg}, we start by establishing the theoretical fundamentals of semi net-form games.  First, we give a formal definition of the semi net-form game.  Secondly, we motivate and define a new game theoretic equilibrium concept called ``level-K relaxed strategies" that can be used to make predictions on a semi net-form game.  Motivated by computational issues, we then present variants of this equilibrium concept that improve both computational efficiency and prediction variance.  In Section~\ref{sec:encounters}, we use a semi net-form game to model the collision avoidance problem and discuss in detail the modeling of a 2-aircraft mid-air encounter.  First, we specify each element of the semi net-form game model and describe how we compute a sample of the game theoretic equilibrium.  Secondly, we describe how to extend the game across time to simulate a complete encounter.  Then we present the results of a sensitivity analysis on the model and examine the potential benefits of a horizontal advisory system.  Finally, we conclude via a discussion of semi net-form game benefits in Section~\ref{sec:benefits} and concluding remarks in Section~\ref{sec:conclusion}.

\section{Semi Network-Form Games}
\label{sec:snfg}

Before we formally define the semi net-form game and various player strategies, we first define the notation used throughout the chapter.

\subsection{Notation}
\label{sec:notation}
Our notation is a combination of standard game theory notation and standard Bayes net notation.
The probabilistic simplex over a space $Z$ is written as $\Delta_Z$. Any Cartesian product $\times_{y \in Y} \Delta_Z$ is written as $\Delta_{Z \mid
Y}$. So $\Delta_{Z \mid Y}$ is the space of all possible conditional
distributions of $z \in Z$ conditioned on a value $y \in Y$. 

We indicate the size of any finite set $A$ as $|A|$. Given a function $g$ with domain $X$ and a subset $Y \subset X$, we write $g(Y)$ to mean the set $\{g(x) : x \in Y\}$. We couch the discussion in terms of countable spaces, but much of the discussion carries over to the uncountable case, e.g., by replacing Kronecker deltas $\delta_{a, b}$ with Dirac deltas $\delta(a - b)$. 

We use uppercase letters to indicate a random variable or its range, with the context making the choice clear.  We use lowercase letters to indicate a particular element of the associated random variable's range, i.e., a particular value of that random variable.  When used to indicate a particular player $i$, we will use the notation $-i$ to denote all players excluding player $i$.  We will also use primes to indicate sampled or dummy variables.

\subsection{Definition}
\label{sec:defSnfg}

A semi net-form game uses a Bayes net to serve as the underlying probabilistic framework, consequently representing all parts of the system using random variables.  Non-human components such as automation and physical systems are described using ``chance" nodes, while human components are described using ``decision" nodes.  Formally, chance nodes differ from decision nodes in that their conditional probability distributions are pre-specified.  Instead each decision node is associated with a utility function, which maps an instantiation of the net to a real number quantifying the player's utility.  To fully specify the Bayes net, it is necessary to determine the conditional distributions at the decision nodes to go with the distributions at the chance nodes.  We will discuss how to arrive at the players' conditional distributions (over possible actions), also called their ``strategies," later in Section~\ref{sec:LKRelaxed}.  We now formally define a semi network-form game as follows:

\begin{definition} 
\label{def:snfg}
An {\bf{($N$-player) semi network-form game}} is a quintuple $(G, X, u, R, \pi)$ where
\begin{enumerate}

\item $G$ is a finite directed acyclic graph $\{V, E\}$, where $V$ is the set of vertices and $E$ is the set of connecting edges of the graph. We write the set of parent nodes of any node $v \in V$ as $pa(v)$ and its successors as $succ(v)$.

\item $X$ is a Cartesian product of $|V|$ separate finite sets, each with at least two elements, with the set for element $v \in V$ written as $X_v$, and the Cartesian product of sets for all elements in $pa(v)$ written as $X_{pa(v)}$.

\item  $u$ is a function $X \rightarrow {\mathbb{R}}^N$. We will typically view it as a set of $N$ utility functions $u_i : X \rightarrow \mathbb{R}$.

\item $R$ is a partition of $V$ into $N+1$ subsets the first $N$ of which have exactly one element. The elements of $R(1)$ through $R(N)$ are called ``Decision" nodes, and the elements of $R(N+1)$ are ``Chance" nodes.  We write $\textsf{D} \equiv \cup_{i=1}^N R(i)$ and $\textsf{C} \equiv R(N+1)$.

\item $\pi$ is a function from
$v \in R(N+1) \rightarrow \Delta_{{X_v} \mid \times_{v' \in
pa(v)}{X_{v'}}}$.
(In other words, $\pi$ assigns to every $v \in
R(N+1)$ a conditional probability distribution of $v$
conditioned on the values of its parents.)
\end{enumerate}
\end{definition}

Intuitively, ${X_v}$ is the set of all possible states at node $v$,
$u_i$ is the utility function of player $i$, $R(i)$ is the decision node set by player $i$, and $\pi$ is the fixed set of
distributions at chance nodes. As an example, a normal-form game~\protect{\cite{MyersonBook}} is a semi net-form game in
which $E$ is empty. As another example, let $v$ be a decision node of 
player $i$ that has one parent, $v'$. Then the conditional distribution $P(X_{v'} \mid X_{pa(v')})$ is a generalization of an information set.  

A semi net-form game is a special case of a general network-form game \protect{\cite{WolpertNFG}}.  In particular, a semi net-form game allows each player to control only one decision node, whereas the full network-form game makes no such restrictions allowing a player to control multiple decision nodes in the net.  Branching (via ``branch nodes") is another feature not available in semi net-form games.  Like a net-form game, Multi-Agent Influence Diagrams \protect{\cite{Koller03}} also allow multiple nodes to be controlled by each player. Unlike a net-form game, however, they do not consider bounded rational agents, and have special utility nodes rather than utility functions.

\subsection{A Simple Semi Network-Form Game Example}
\label{sec:simpleNFGExample}

We illustrate the basic understandings of semi net-form games using the simple example shown in Figure~\ref{fig:simpleNFG}.  In this example, there are 6 random variables ($A, B, C, D, P_1, P_2$) represented as nodes in the net; the edges between nodes define the conditional dependence between random variables.  For example, the probability of $D$ depends on the values of $P_1$ and $P_2$, while the probability of $A$ does not depend on any other variables.  We distinguish between the two types of nodes: chance nodes ($A, B, C, D$), and decision nodes ($P_1, P_2$).  As discussed previously,  chance nodes differ from decision nodes in that their conditional probability distributions are specified a-priori.  Decision nodes do not have these distributions pre-specified, but rather what is pre-specified are the utility functions ($U_1$ and $U_2$) of those players.  Using their utility functions, their strategies $P(P_1 \mid B)$ and $P(P_2 \mid C)$ are computed to complete the Bayes net.  This computation requires the consideration of the Bayes net from each player's perspective.

\begin{figure}
\begin{center}
\includegraphics[width = \textwidth, trim = 0 80 80 80]{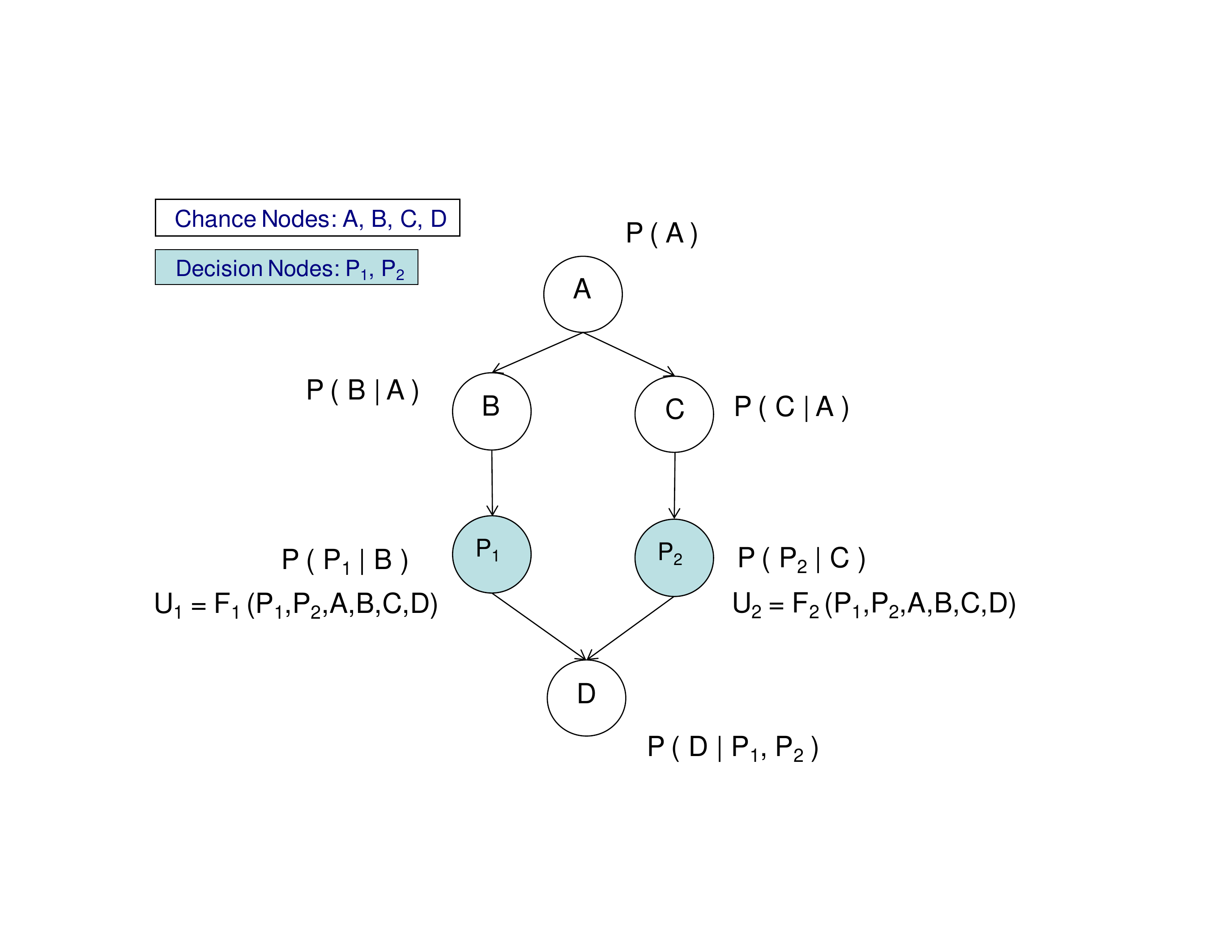}
\caption{A simple net-form game example:  Fixed conditional probabilities are specified for chance nodes ($A,B,C,D$), while utility functions are specified for decision nodes ($P_1, P_2$).  Players try to maximize their expected utility over the Bayes net.}
\label{fig:simpleNFG}       
\end{center}
\end{figure}

Figure~\ref{fig:simpleNFGPlayer} illustrates the Bayes net from $P_1$'s perspective.  In this view, there are nodes that are observed ($B$), there are nodes that are controlled ($P_1$), and there are nodes that do not fall into any of these categories ($A, C, P_2, D$), but appear in the player's utility function.  This arises from the fact that in general the player's utility function can be a function of any variable in the net.  As a result, in order to evaluate the expected value of his utility function for a particular candidate action (sometimes we will use the equivalent game theoretic term ``move"), $P_1$ must perform inference over these variables based on what he observes\footnote{We discuss the computational complexity of a particular equilibrium concept later in Section~\ref{sec:complexity}.}.  Finally, the player chooses the action that gives the highest expected utility.

\begin{figure}
\begin{center}
\includegraphics[width=\textwidth,trim = 30 100 30 80]{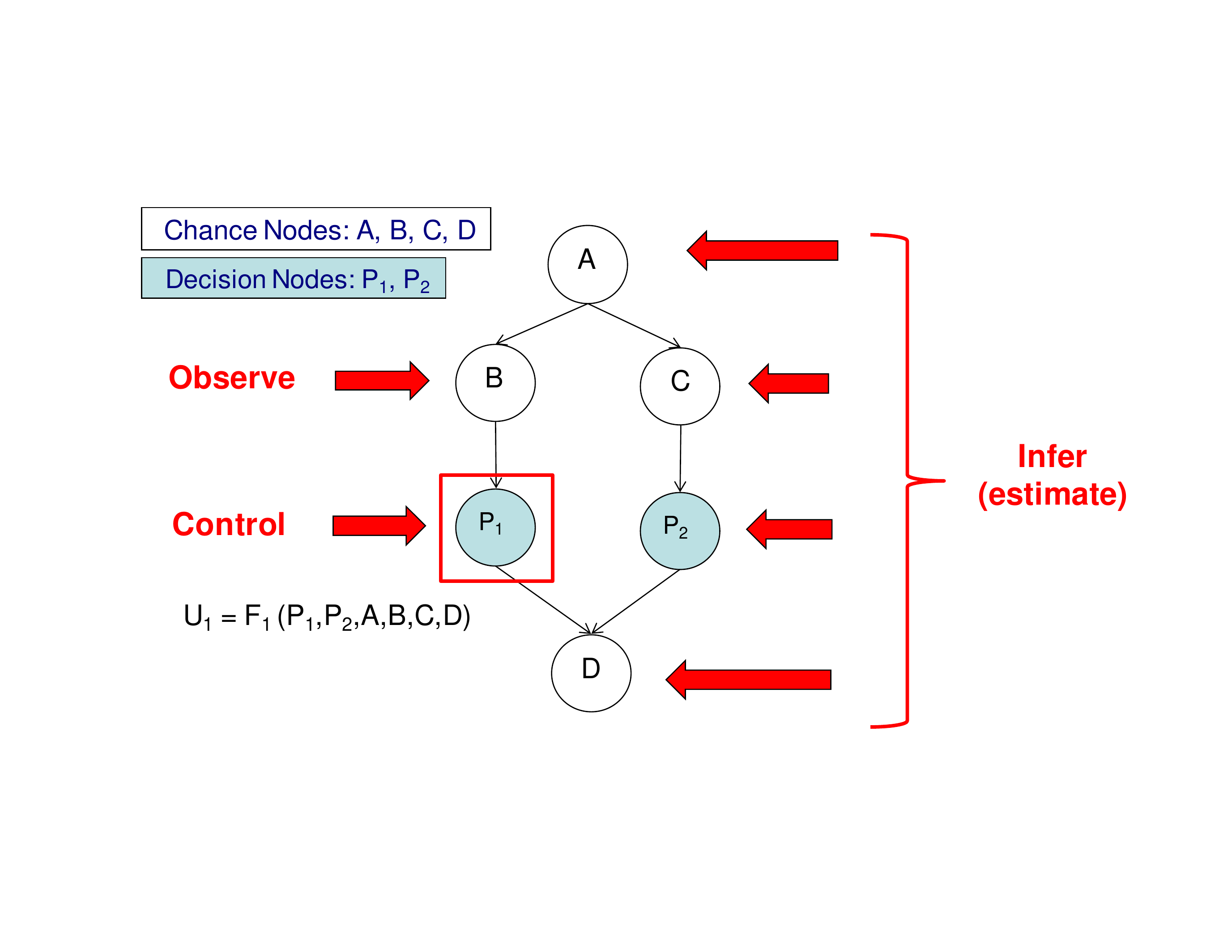}
\caption{A simple net-form example game from player 1's perspective:  Using information that he observes, the player infers over unobserved variables in the Bayes net in order to set the value of his decision node.}
\label{fig:simpleNFGPlayer}       
\end{center}
\end{figure}

\subsection{Level-K Thinking}
\label{sec:LK}

Level-K thinking is a game theoretic equilibrium concept used to predict the outcome of human-human interactions.  A number of studies \protect{\cite{Camerer10,Camerer89,CamererBook, CostaGomes09, Crawford07,Wright10}} have shown promising results predicting experimental data in games using this method.  The concept of level-K is defined recursively as follows. A level $K$ player plays (picks his action) as though all other players are playing at level $K-1$, who, in turn, play as though all other players are playing at level $K-2$, etc.  The process continues until level 0 is reached, where the player plays according to a prespecified prior distribution.  Notice that running this process for a player at $K \ge 2$ results in ricocheting between players.  For example, if player A is a level 2 player, he plays as though player B is a level 1 player, who in turn plays as though player A is a level 0 player.  Note that player B in this example may not be a level 1 player in reality -- only that player A assumes him to be during his reasoning process.  Since this ricocheting process between levels takes place entirely in the player's mind, no wall clock time is counted (we do not consider the time it takes for a human to run through his reasoning process).  We do not claim that humans actually think in this manner, but rather that this process serves as a good model for predicting the outcome of interactions at the aggregate level.  In most games, $K$ is a fairly low number for humans; experimental studies \protect{\cite{CamererBook}} have found $K$ to be somewhere between 1 and 2.

Although this work uses level-K exclusively, we are by no means wedded to this equilibrium concept.  In fact, semi net-form games can be adapted to use other models, such as Nash equilibrium, Quantal Response Equilibrium, Quantal Level-K, and Cognitive Hierarchy.  Studies \protect{\cite{CamererBook,Wright10}} have found that performance of an equilibrium concept varies a fair amount depending on the game.  Thus it may be wise to use different equilibrium concepts for different problems.

\subsection{Satisficing}
\label{sec:BoM}

Bounded rationality as coined by Simon \protect{\cite{Simon56}} stems from observing the limitations of humans during the decision-making process. That is, humans are limited by the information they have, cognitive limitations of their minds, and the finite amount of time they have to make decisions. The notion of satisficing \protect{\cite{Caplin10,Simon56,Simon82}} states that humans are unable to evaluate the probability of all outcomes with sufficient precision, and thus often make decisions based on adequacy rather than by finding the true optimum.  Because decision-makers lack the ability and resources to arrive at the optimal solution, they instead apply their reasoning only after having greatly simplified the choices available.

Studies have shown evidence of satisficing in human decision-making.  In recent experiments \protect{\cite{Caplin10}}, subjects were given a series of calculations (additions and subtractions), and were told that they will be given a monetary prize equal to the answer of the calculation that they choose. Although the calculations were not difficult in nature, they did take effort to perform.  The study found that most subjects did not exhaustively perform all the calculations, but instead stopped when a ``high enough" reward was found.  

\subsection{Level-K Relaxed Strategies}
\label{sec:LKRelaxed}

We use the notions of level-K thinking and satisficing to motivate a new game theoretic equilibrium concept called ``level-K relaxed strategies."  For a player $i$ to perform classic level-K reasoning~\protect{\cite{CamererBook}} requires $i$
to calculate best responses\footnote{We use the term best response in the game theoretic sense.  i.e. the player chooses the move with the highest expected utility.}. In turn, calculating best responses often involves calculating the Bayesian posterior probability of what information is available to the other players, $-i$, conditioned on the information available to $i$. That posterior is an integral, which typically cannot be evaluated in closed form.

In light of this, to use level-K reasoning, players must approximate
those Bayesian integrals. We hypothesize that real-world players do
this using Monte Carlo sampling.  Or more precisely, we hypothesize
that their behavior is consistent with their approximating the
integrals that way. 

More concretely, given a node $v$, to form their best-response, the associated player $i = R^{-1}(v)$ will want to calculate quantities of the form argmax$_{x_v} [\mathbb{E}(u_i \mid x_v, x_{pa(v)})]$, where $u_i$ is the player's utility, $x_v$ is the variable set by the player (i.e. his move), and $x_{pa(v)}$ is the realization of his parents that he observes.  We hypothesize that he (behaves as though he) approximates this calculation in several steps. First, $M$ candidate moves are chosen via IID sampling the player's satisficing distribution.  Now, for each candidate move, he must estimate the expected utility resulting from playing that move.  He does this by sampling the posterior probability distribution $P^K(X_V \mid x_v, x_{pa(v)})$ (which accounts for what he knows), and computing the sample expectation $\hat{u}^K_i$.  Finally, he picks the move that has the highest estimated expected utility.  Formally, we give the following definition:

\begin{definition}
\label{def:lk_symbols}
Consider a semi network-form game $(G,X,u,R,\pi)$ with level $K-1$ relaxed strategies\footnote{We will define level-K relaxed strategies in Definition~\ref{def:lk_1}.} $\Lambda^{K-1}(X_{v'} \mid X_{pa(v')})$ defined for all $v' \in \textsf{D}$ and $K \ge 1$.  For all nodes $v$ and sets of nodes $Z$ in such a semi net-form game, define 
\begin{enumerate}
\item
$U = V \setminus \{v, pa(v)\}$, 
\item
$P^K(X_{v} \mid X_{pa(v)}) = \pi(X_{v} \mid X_{pa(v)})$ if $v \in \textsf{C}$,
\item
$P^K(X_{v} \mid X_{pa(v)}) =  \Lambda^{K-1}(X_{v} \mid X_{pa(v)})$ if $v \in \textsf{D}$, and
\item
$P^K(X_Z) = \prod_{v'' \in Z} P^K(X_{v''} \mid X_{pa(v'')})$.
\end{enumerate}
\end{definition}

\begin{definition} Consider a semi network-form game
$(G, X, u, R, \pi)$. For all $v \in \textsf{D}$, specify an associated
\textbf{level 0} distribution $\Lambda^0(X_{v} \mid x_{pa(v)}) \in
\Delta_{{X_{v}} \mid \times_{v' \in pa(v)}{X_{v'}}}$ and an associated
\textbf{satisficing}  distribution $\lambda(X_v \mid x_{pa(v)}) \in
\Delta_{{X_v} \mid \times_{v' \in pa(v)}{X_{v'}}}$. Also specify
counting numbers $M$ and $M'$.

For any $K \ge 1$, the \textbf{level $K$ relaxed strategy} of
node $v \in \textsf{D}$
is the conditional distribution $\Lambda^K(X_v \mid x_{pa(v)}) \in
\Delta_{{X_v} \mid \times_{v' \in pa(v)}{X_{v'}}}$ sampled by 
running the following stochastic process independently for each $x_{pa(v)} \in X_{pa(v)}$:

\begin{enumerate}
\item Form a set $\{x'_v(j) : j = 1, \ldots, M\}$ of IID samples of $\lambda(X_v \mid x_{pa(v)})$ and then remove all duplicates.  Let $m$ be the resultant size of the set;

\item For $j = 1,\dots,m$, form a set $\{x'_V(k; x'_v(j)) : k =
1,\ldots M'\}$ of IID samples of the joint distribution
\begin{eqnarray*}
P^K(X_V \mid x'_v(j), x_{pa(v)}) &=& \prod_{v' \in V} P^K(X_{v'} \mid X_{pa(v')}) \delta_{X_{pa(v)},
x_{pa(v)}} \delta_{X_v, x'_v(j)};
\end{eqnarray*}
and compute
\begin{eqnarray*}
\hat{u}^K_i(x'_U(;x'_v(j)),x'_v(j),x_{pa(v)}) &=& \frac{1}{M'} \sum_{k = 1}^{M'} u_i(x'_V(k,x'_v(j)));
\end{eqnarray*}
where $x'_V(;x'_v(j))$ is shorthand for $\{x'_{v'}(k,x'_v(j)) : v' \in V, k = 1, \ldots, M'\}$

\item  Return $x'_v(j^*)$ where $j^* \equiv$ argmax$_{j} [\hat{u}^K_i(x'_U(;x'_v(j)),x'_v(j),x_{pa(v)})]$.
\end{enumerate}
\label{def:lk_1}
\end{definition}

Intuitively, the counting numbers $M$ and $M'$ can be interpreted as a measure of a player's rationality.  Take, for example, $M \rightarrow \infty$ and $M' \rightarrow \infty$.  Then the player's entire movespace would be considered as candidate moves, and the expected utility of each candidate move would be perfectly evaluated.  Under these circumstances, the player will always choose the best possible move, making him perfectly rational.  On the other hand if $M=1$, this results in the player choosing his move according to his satisficing distribution, corresponding to random behavior.

One of the strengths of Monte Carlo expectation estimation is that it is unbiased \protect{\cite{RobertBook}}. 
This property carries over to level-K relaxed strategies. More precisely,
consider a level $K$ relaxed player $i$, deciding
which of his moves $\{x'_v(j) :  j \in 1, \ldots, m\}$ to play
for the node $v$ he controls, given a particular set of values $x_{pa(v)}$ that he observes.
To do this he will compare the values $\hat{u}^K_i(x'_U(;x'_v(j)),x'_v(j),x_{pa(v)})$.
These values are all unbiased estimates of the associated conditional expected utility\footnote{Note that the true expected conditional utility is not defined without an associated complete Bayes net.  However, we show in Theorem~\ref{thm:relaxed} Proof that the expected conditional utility is actually independent of the probability $P_{\Gamma_i}(X_v \mid X_{pa(v)})$ and so it can chosen arbitrarily. We make the assumption that $P_{\Gamma_i}(x_v \mid x_{pa(v)}) \ne 0$ for mathematical formality to avoid dividing by zero in the proof.} evaluated under an equivalent Bayes Net $\Gamma_i$ defined in Theorem~\ref{thm:relaxed}.  Formally, we have the following:

\begin{theorem}
\label{thm:relaxed}
Consider a semi net-form game $(G, X, u, R, \pi)$ with associated satisficing $\lambda(X_{v} \mid x_{pa(v)})$ and level 0 distribution $\Lambda^0(X_{v} \mid x_{pa(v)})$ specified for all players.

Choose a particular player $i$ of that game, a particular level $K$, and a player move $x_v=x'_v(j)$ from Definition~\ref{def:lk_1} for some particular $j$.  Consider any values $x_{pa(v)}$ where $v$ is the node controlled by player $i$.  Define $\Gamma_i$ as any Bayes net whose directed acyclic graph is $G$, where for all nodes $v' \in \textsf{C}$, $P_{\Gamma_i}(X_{v'} \mid X_{pa(v')})=\pi(X_{v'} \mid X_{pa(v')})$, for all nodes $v'' \in \textsf{D}$, $P_{\Gamma_i}(X_{v''} \mid X_{pa(v'')})$, and where $P_{\Gamma_i}(X_v \mid X_{pa(v)})$ is arbitrary so long as $P_{\Gamma_i}(x_v \mid x_{pa(v)}) \ne 0$.  We also define the notation $P_{\Gamma_i}(X_Z)$ for a set Z of nodes to mean $\prod_{v' \in Z} P_{\Gamma_i}(X_{v'} \mid X_{pa(v')})$.

Then
the expected value $\mathbb{E}(\hat{u}^K_i \mid x'_v(j), x_{pa(v)})$ 
evaluated under the associated level-K relaxed strategy equals $\mathbb{E}(u_i \mid x_v, x_{pa(v)})$ evaluated under the Bayes net $\Gamma_i$.
\end{theorem}

\begin{proof}
Write
\begin{eqnarray*}
&&\mathbb{E}(\hat{u}^K_i \mid x'_v(j), x_{pa(v)}) \;= \\
&&\int dx'_V(;x'_v(j)) \; P(x'_V(;x'_v(j)) \mid  x'_v(j), x_{pa(v)}) \hat{u}^K_i(x'_V(;x'_v(j))) \\
&=& \int dx'_V(;x'_v(j)) \; P(x'_V(;x'_v(j)) \mid x'_v(j), x_{pa(v)}) \frac{1}{M'} \sum_{k = 1}^{M'} u_i(x'_V(k,x'_v(j))) \\
&=& \frac{1}{M'} \sum_{k = 1}^{M'} \int dx'_V(k,x'_v(j)) \; P^K(x'_V(k,x'_v(j)) \mid x'_v(j), x_{pa(v)}) u_i(x'_V(k,x'_v(j))) \\
&=& \frac{1}{M'} \sum_{k = 1}^{M'} \int dX_V \; P^K(X_V \mid x_v, x_{pa(v)}) u_i(X_U,x_v,x_{pa(v)}) \\
&=& \int dX_V \; P^K(X_V \mid x_v, x_{pa(v)}) u_i(X_U,x_v,x_{pa(v)}) \\
&=& \frac{\int dX_U \; P^K(X_U,x_v, x_{pa(v)}) u_i(X_U,x_v,x_{pa(v)})}{\int dX_U \; P^K(X_U, x_v, x_{pa(v)})} \\
&=& \frac{\int dX_U \; \prod_{v' \in U} P^K(X_{v'} \mid X_{pa(v')}) \prod_{v'' \in pa(v)} P^K(x_{v''} \mid X_{pa(v'')}) P^K(x_v \mid x_{pa(v)}) u_i(X_U,x_v,x_{pa(v)})}{\int dX_U \;  \prod_{z' \in U} P^K(X_{z'} \mid X_{pa(z')}) \prod_{z'' \in pa(v)} P^K(x_{z''} \mid X_{pa(z'')}) P^K(x_v \mid x_{pa(v)})}\\
&=& \frac{\int dX_U \; \prod_{v' \in U} P^K(X_{v'} \mid X_{pa(v')}) \prod_{v'' \in pa(v)} P^K(x_{v''} \mid X_{pa(v'')}) u_i(X_U,x_v,x_{pa(v)})}{\int dX_U \;  \prod_{z' \in U} P^K(X_{z'} \mid X_{pa(z')}) \prod_{z'' \in pa(v)} P^K(x_{z''} \mid X_{pa(z'')})}\\
&=& \frac{\int dX_U \; \prod_{v' \in U} P_{\Gamma_i}(X_{v'} \mid X_{pa(v')}) \prod_{v'' \in pa(v)} P_{\Gamma_i}(x_{v''} \mid X_{pa(v'')})u_i(X_U,x_v,x_{pa(v)})}{\int dX_U \;  \prod_{z' \in U} P_{\Gamma_i}(X_{z'} \mid X_{pa(z')}) \prod_{z'' \in pa(v)} P_{\Gamma_i}(x_{z''} \mid X_{pa(z'')})}\\
&=& \int dX_V \; P_{\Gamma_i}(X_V \mid x_v, x_{pa(v)}) u_i(X_U,x_v,x_{pa(v)}) \\
&=& \mathbb{E}(u_i \mid x_v, x_{pa(v)}) 
\qquad \qquad \qquad \qquad \qquad \qquad \qquad \qquad \qquad \qquad \qquad \qquad \qedsymbol
\end{eqnarray*}
\end{proof}

In other words, we can set $P^K(x_v \mid x_{pa(v)})$ arbitrarily (as long as it is nonzero) and still have the utility estimate evaluated under the associated level-K relaxed strategy be an unbiased estimate of the expected utility conditioned on $x_v$ and $x_{pa(v)}$ evaluated under $\Gamma_i$.  Unbiasness in level-K relaxed strategies is important because the player must rely on a limited number of samples to estimate the expected utility of each candidate move.  The difference of two unbiased estimates is itself unbiased, enabling the player to compare estimates of expected utility without bias.

\subsection{Level-K d-Relaxed Strategies}
\label{sec:LKDRelaxed}

A practical problem with relaxed strategies is that the number of
samples may grow very quickly with depth of the Bayes net.
The following example illustrates another problem:

\begin{example}
Consider a net form game with two simultaneously moving players,
Bob and Nature, both making $\mathbb{R}$-valued moves.  Bob's utility
function is given by the difference between his and Nature's move. 

So to determine his level 1 relaxed strategy, Bob chooses $M$ candidate moves by sampling his
satisficing distribution, and then Nature chooses $M'$ (``level 0'')
moves for each of those $M$ moves by Bob. In truth, one of Bob's $M$ candidate
moves, $x_{Bob}^*$, is dominant\footnote{We use the term dominant in the game theoretic sense.  i.e., the move $x_{Bob}^*$ gives Bob the highest expected utility no matter what move Nature makes.} over the other $M - 1$ candidate moves due to
the definition of the utility function. However since there are an
independent set of $M'$ samples of Nature for each of Bob's moves,
there is nonzero probability that Bob won't return $x_{Bob}^*$, i.e.,
his level 1 relaxed strategy has nonzero probability of returning some other move.
\label{ex:5}
\end{example}

As it turns out, a slight modification to the Monte Carlo process defining 
relaxed strategies results in Bob returning $x_{Bob}^*$ with probability $1$
in Example~\ref{ex:5} for many graphs $G$. 
This modification also reduces the explosion in the number of Monte Carlo samples required for computing the players' strategies.

This modified version of relaxed strategies works by setting aside
a set $Y$ of nodes which are statistically independent of the state
of $v$. Nodes in $Y$ do not have to be
resampled for each value $x_v$.  Formally, the set $Y$ will be defined using the 
dependence-separation (d-separation) property concerning the
groups of nodes in the graph $G$ that defines the semi net-form
game~\protect{\cite{KollerBook,Koller03,Pearl00}}. Accordingly, we call this
modification ``d-relaxed strategies."  
Indeed, by \emph{not} doing any
such resampling, we can exploit the ``common random
numbers" technique to improve the Monte Carlo estimates~\protect{\cite{RobertBook}}.  Loosely speaking, to choose the move with the highest estimate of expected utility requires one to compare all pairs of estimates and thus implicitly evaluate their differences.  Recall that the variance of a difference of two estimates is given by $Var(\chi-\upsilon) = Var(\chi) + Var(\upsilon) - 2Cov(\chi,\upsilon)$.  By using d-relaxed strategies, we expect the covariance $Cov(\chi,\upsilon)$ to be positive, reducing the overall variance in the choice of the best move.

\begin{definition}
\label{def:rejSamp_symbols}
Consider a semi network-form game $(G,X,u,R,\pi)$ with level $K-1$ d-relaxed strategies\footnote{We will define level-K d-relaxed strategies in Definition~\ref{def:rejectionSampling}.} $\bar{\Lambda}^{K-1}(X_{v'} \mid X_{pa(v')})$ defined for all $v' \in \textsf{D}$ and $K \ge 1$.  For all nodes $v$ and sets of nodes $Z$ in such a semi net-form game, 
define 
\begin{enumerate}
\item
$S^{v} = succ(v)$, 
\item
$S^{-v} = V \setminus \{v \cup S^{v}\}$, 
\item
$Y = V \setminus \{v \cup pa(v) \cup S^{v} \}$,
\item
$\bar{P}^K(X_{v} \mid X_{pa(v)}) = \pi(X_{v} \mid X_{pa(v)})$ if $v \in \textsf{C}$,
\item
$\bar{P}^K(X_{v} \mid X_{pa(v)}) =  \bar{\Lambda}^{K-1}(X_{v} \mid X_{pa(v)})$ if $v \in \textsf{D}$, and
\item
$\bar{P}^K(X_Z) = \prod_{v'' \in Z} \bar{P}^K(X_{v''} \mid X_{pa(v'')})$.
\end{enumerate}
\end{definition}

Note that $Y \cup pa(v) = S^{-v}$ and $v \cup S^{v} \cup S^{-v} = V$.  The motivation for these definitions comes from the fact that $Y$ is precisely the set of nodes that are d-separated from $v$ by $pa(v)$.  As a result, when the player who controls $v$ samples $X_v$ conditioned on the observed $x_{pa(v)}$, the resultant value $x_v$ is statistically independent of the values of all the nodes in $Y$. Therefore the same set of samples of the values of the nodes in $Y$ can be reused for each new sample of $X_v$. This kind of reuse can provide substantial computational savings in
the reasoning process of the player who controls $v$.  We now consider the modified sampling process noting that a level-K d-relaxed strategy is defined recursively in $K$, via the sampling of $\bar{P}^K$.  Note that in general, Definition~\ref{def:lk_1} and Definition~\ref{def:rejectionSampling} do not lead to the same player strategies (conditional distributions) as seen in Example~\ref{ex:5}.

\begin{definition} 
\label{def:rejectionSampling}
Consider a semi network-form game
$(G, X, u, R, \pi)$ with associated
level 0 distributions $\Lambda^0(X_{v} \mid x_{pa(v)})$ and
satisficing distributions $\lambda(X_v \mid x_{pa(v)})$. Also specify counting numbers $M$ and $M'$. 

For any $K \ge 1$, the \textbf{level $K$ d-relaxed strategy} of
node $v \in \textsf{D}$, where $v$ is controlled by player $i$, 
is the conditional distribution $\bar{\Lambda}^K(X_v \mid x_{pa(v)}) \in
\Delta_{{X_v} \mid \times_{v' \in pa(v)}{X_{v'}}}$ that is sampled by 
running the following stochastic process independently for each $x_{pa(v)} \in X_{pa(v)}$:
\begin{enumerate}
\item Form a set $\{x'_v(j) : j = 1, \ldots, M\}$ of IID samples of $\lambda(X_v \mid x_{pa(v)})$ and then remove all duplicates.  Let $m$ be the resultant size of the set;

\item  Form a set $\{x'_{S^{-v}}(k) : k = 1,\ldots, M'\}$ of IID samples of the distribution over $X_{S^{-v}}$ given by
\begin{eqnarray*}
\bar{P}^K(X_{S^{-v}} \mid x_{pa(v)}) &=&  \prod_{v' \in {S^{-v}} } \bar{P}^K(X_{v'} \mid X_{pa(v')}) \delta_{X_{pa(v)},x_{pa(v)}};
\end{eqnarray*}

\item For $j = 1,\dots,m$, form a set $\{x'_{S^v}(k, x'_v(j)) : k =
1,\ldots, M'\}$ of IID samples of the distribution over $X_{S^v}$
given by
\begin{eqnarray*}
\bar{P}^K(X_{S^v} \mid x'_Y(;),x'_v(j), x_{pa(v)}) &=& \prod_{v' \in S^v} \bar{P}^K(X_{v'} \mid X_{pa(v')}) \prod_{v'' \in S^{-v}} 
\delta_{X_{v''}, x'_{v''}(k)} 
\delta_{X_{v}, x'_{v}(j)};
\end{eqnarray*}

and compute
\begin{eqnarray*}
&& \bar{u}^K_i(x'_Y(;),x'_{S^v}(;x'_v(j)),x'_v(j),x_{pa(v)}) \; = \\
&& \qquad \qquad \qquad \qquad \qquad \qquad \frac{1}{M'} \sum_{k = 1}^{M'} u_i(x'_Y(k),x'_{S^v}(k,x'_v(j)),x'_v(j),x_{pa(v)});
\end{eqnarray*}
where $x'_{Y}(;)$ is shorthand for $\{x'_v(k) : v \in Y, k = 1, \ldots, M'\}$
and $x'_{S^v}(;x'_v(j))$ is shorthand for $\{x'_{S^v}(k, x'_v(j)) : k = 1,\ldots, M'\}$.

\item  Return $x'_v(j^*)$ where $j^* \equiv$ argmax$_{j} [\bar{u}^K_i(x'_Y(;),x'_{S^v}(;x'_v(j)),x'_v(j),x_{pa(v)})]$.
\end{enumerate}

\end{definition}

Definition~\ref{def:rejectionSampling} requires directly sampling from a conditional probability, which requires rejection sampling.  This is highly inefficient if $pa(v)$ has low probability, and actually impossible if $pa(v)$ is continuous.  For these computational considerations, we introduce a variation of the previous algorithm based on likelihood-weighted sampling rather than rejection sampling.  Although the procedure, as we shall see in Definition~\ref{def:unnormLwSampling}, is only able to estimate the player's expected utility up to a proportionality constant (due to the use of likelihood-weighted sampling \protect{\cite{KollerBook}}), we point out that this is sufficient since proportionality is all that is required to choose between candidate moves.  Note that un-normalized likelihood-weighted level-K d-relaxed strategy, like level-K d-relaxed strategy, is defined recursively in $K$.

\begin{definition}
Consider a semi network-form game $(G,X,u,R,\pi)$ with unnormalized likelihood-weighted level $K-1$ d-relaxed strategies\footnote{We will define unnormalized likelihood-weighted level-K d-relaxed strategies in Definition~\ref{def:unnormLwSampling}.} $\tilde{\Lambda}^{K-1}(X_{v'} \mid X_{pa(v')})$ defined for all $v' \in \textsf{D}$ and $K \ge 1$.  For all nodes $v$ and sets of nodes $Z$ in such a semi net-form game, 
define 
\begin{enumerate}
\item
$\tilde{P}^K(X_{v} \mid X_{pa(v)}) = \pi(X_{v} \mid X_{pa(v)})$ if $v \in \textsf{C}$,
\item
$\tilde{P}^K(X_{v} \mid X_{pa(v)}) =  \tilde{\Lambda}^{K-1}(X_{v} \mid X_{pa(v)})$ if $v \in \textsf{D}$, and
\item
$\tilde{P}^K(X_Z) = \prod_{v'' \in Z} \tilde{P}^K(X_{v''} \mid X_{pa(v'')})$.
\end{enumerate}
\end{definition}

\begin{definition}
\label{def:unnormLwSampling}
Consider a semi network-form game
$(G, X, u, R, \pi)$ with associated
level 0 distributions $\Lambda^0(X_{v} \mid x_{pa(v)})$ and
satisficing distributions $\lambda(X_v \mid x_{pa(v)})$. Also specify counting numbers $M$ and $M'$, and recall the meaning of set $Y$ from Definition~\ref{def:rejSamp_symbols}.

For any $K \ge 1$, the \textbf{un-normalized likelihood-weighted level $K$ d-relaxed strategy} of
node $v$, where node $v$ is controlled by player $i$, 
is the conditional distribution $\tilde{\Lambda}^K(X_v \mid x_{pa(v)}) \in
\Delta_{{X_v} \mid \times_{v' \in pa(v)}{X_{v'}}}$ that is sampled by 
running the following stochastic process independently for each $x_{pa(v)} \in X_{pa(v)}$:
\begin{enumerate}
\item Form a set $\{x'_v(j) : j = 1, \ldots, M\}$ of IID samples of $\lambda(X_v \mid x_{pa(v)})$, and then remove all duplicates.  Let $m$ be the resultant size of the set;

\item  Form a set of weight-sample pairs $\{(w'(k),x'_{S^{-v}}(k)) : k =
1,\ldots M'\}$ by setting $x'_{pa(v)} = x_{pa(v)}$, IID sampling the distribution over $X_Y$ given by
\begin{eqnarray*}
\tilde{P}^K(X_Y) &=&  \prod_{v' \in Y } \tilde{P}^K(X_{v'} \mid X_{pa(v')})
\end{eqnarray*}
and setting
\begin{eqnarray*}
w'(k) &=& \prod_{v' \in pa(v)} \tilde{P}^K(x_{v'} \mid x'_{pa(v')}(k));
\end{eqnarray*}

\item For $j = 1,\dots,m$, form a set $\{x'_{S^v}(k, x'_v(j)) : k =
1,\ldots M'\}$ of IID samples of the distribution over $X_{S^v}$
given by
\begin{eqnarray*}
\tilde{P}^K(X_{S^v} \mid x'_Y(;),x'_v(j), x_{pa(v)}) &=& \prod_{v' \in S^v} \tilde{P}^K(X_{v'} \mid X_{pa(v')}) \prod_{v'' \in S^{-v}} 
\delta_{X_{v''}, x'_{v''}(k)} 
\delta_{X_{v}, x'_{v}(j)};
\end{eqnarray*}

and compute
\begin{eqnarray*}
&&\tilde{u}_i(x'_Y(;),x'_{S^v}(;x'_v(j)),x'_v(j),x_{pa(v)}) \; = \\ 
&& \qquad \qquad \qquad \qquad \frac{1}{M'} \sum_{k = 1}^{M'} w'(k) u_i(x'_Y(k),x'_{S^v}(k,x'_v(j)),x'_v(j),x_{pa(v)});
\end{eqnarray*}

\item  Return $x'_v(j^*)$ where $j^* \equiv$ argmax$_{j} [\tilde{u}_i(x'_Y(k),x'_{S^v}(k,x'_v(j)),x'_v(j),x_{pa(v)})]$.
\end{enumerate}

\end{definition}


\subsubsection{Computational Complexity}
\label{sec:complexity}
Let $N$ be the number of players.  Intuitively, as each level $K$ player samples  the Bayes net from their perspective, they initiate samples by all other players at level $K-1$.  These players, in turn, initiate samples by all other players at level $K-2$, continuing until level 1 is reached (since level 0 players do not sample the Bayes net).

As an example, Figure~\ref{fig:complexity} enumerates the number of Bayes net samples required to perform level-K d-relaxed sampling for $N=3$ where all players reason at $K=3$.  Each square represents performing the Bayes net sampling process once.  As shown in the figure, the sampling process of $P_A$ at level 3 initiates sampling processes in the two other players, $P_B$ and $P_C$, at level 2.  This cascading effect continues until level 1 is reached, and is repeated from the top for $P_B$ and $P_C$ at level 3.  In general, when all players play at the same level $K$, this may be conceptualized as having $N$ trees of degree $N-1$ and depth $K$; therefore having a computational complexity proportional to $\sum_{j=0}^{K-1}(N-1)^j N$, or $O(N^K)$.  In other words, the computational complexity is polynomial in the number of players and exponential in the number of levels.  Fortunately, experiments \protect{\cite{CamererBook, CostaGomes09}} have found $K$ to be small in human reasoning.  

\begin{figure}
\begin{center}
\includegraphics[width=\textwidth,trim = 0 180 10 180]{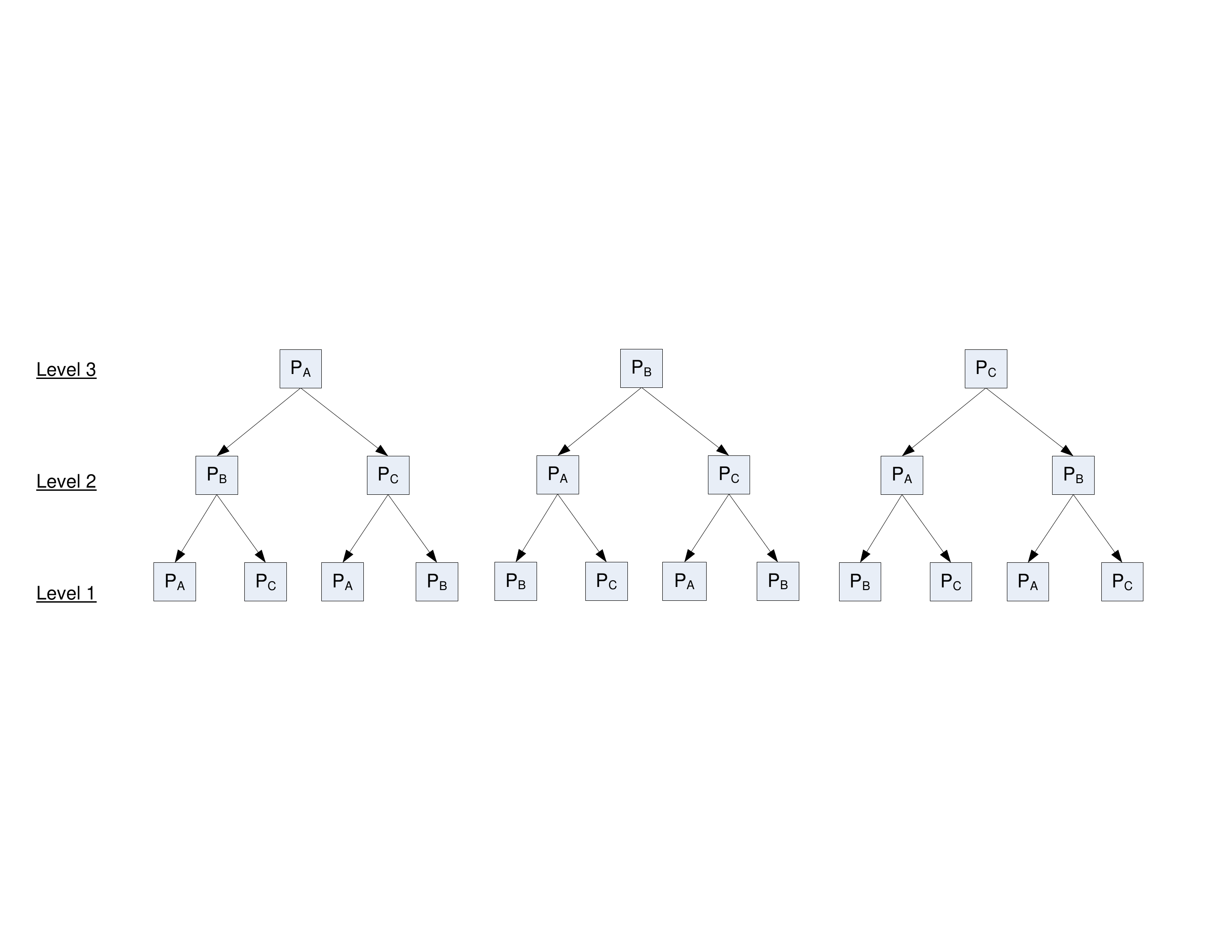}
\caption{Computational complexity of level-K d-relaxed strategies with $N=3$ and $K=3$:  Each box represents a single execution of the algorithm.  The computational complexity is found to be $O(N^K)$.}
\label{fig:complexity}       
\end{center}
\end{figure}

\section{Using Semi Net-Form Games to Model Mid-Air Encounters}
\label{sec:encounters}

\ac{tcas} is an aircraft collision avoidance system currently mandated by the International Civil Aviation Organization to be fitted to all aircraft with a maximum take-off mass of over 5700 kg (12,586 lbs) or authorized to carry more than 19 passengers.  It is an onboard system designed to operate independently of ground-based air traffic management systems to serve as the last layer of safety in the prevention of mid-air collisions.  \ac{tcas} continuously monitors the airspace around an aircraft and warns pilots of nearby traffic.  If a potential threat is detected, the system will issue a \ac{ra}, i.e., recommended escape maneuver, to the pilot.  The \ac{ra} is presented to the pilot in both a visual and audible form.  Depending on the aircraft, visual cues are typically implemented on either an instantaneous vertical speed indicator, a vertical speed tape that is part of a primary flight display, or using pitch cues displayed on the primary flight display.  Audibly, commands such as ``Climb, Climb!" or ``Descend, Descend!" are heard.

If both (own and intruder) aircraft are \ac{tcas}-equipped, the issued \acp{ra} are coordinated, i.e., the system will recommend different directions to the two aircraft.  This is accomplished via the exchange of ``intents" (coordination messages).  However, not all aircraft in the airspace are \ac{tcas}-equipped, i.e., general aviation.  Those that are not equipped cannot issue \acp{ra}.

While \ac{tcas} has performed satisfactorily in the past, there are many limitations to the current \ac{tcas} system.  First, since \ac{tcas} input data is very noisy in the horizontal direction, issued \acp{ra} are in the vertical direction only, greatly limiting the solution space.  Secondly, \ac{tcas} is composed of many complex deterministic rules, rendering it difficult for authorities responsible for the maintenance of the system (i.e., Federal Aviation Administration) to understand, maintain, and upgrade.  Thirdly, \ac{tcas} assumes straight-line aircraft trajectories and does not take into account flight plan information.  This leads to a high false-positive rate, especially in the context of closely-spaced parallel approaches.

This work focuses on addressing one major weakness of \ac{tcas}: the design assumption of a deterministic pilot model.  Specifically, \ac{tcas} assumes that a pilot receiving an \ac{ra} will delay for 5 seconds, and then accelerate at 1/4 g to execute the \ac{ra} maneuver precisely.  Although pilots are trained to obey in this manner, a recent study of the Boston area \protect{\cite{Kuchar07}} has found that only 13\% of \acp{ra} are obeyed -- the aircraft response maneuver met the \ac{tcas} design assumptions in vertical speed and promptness.  In 64\% of the cases, pilots were in partial compliance -- the aircraft moved in the correct direction, but did not move as promptly or as aggressively as instructed.  Shockingly, the study found that in 23\% of \acp{ra}, the pilots actually responded by maneuvering the aircraft in the opposite direction of that recommended by \ac{tcas} (a number of these cases of non-compliance may be attributed to visual flight rules\footnote{Visual flight rules are a set of regulations set forth by the Federal Aviation Administration which allow a pilot to operate an aircraft relying on visual observations (rather than cockpit instruments).}).  As air traffic density is expected to double in the next 30 years~\protect{\cite{FAA10}}, the safety risks of using such a system will increase dramatically.

Pilot interviews have offered many insights toward understanding these statistics.  The main problem is a mismatch between the pilot model used to design the \ac{tcas} system and the behavior exhibited by real human pilots. During a mid-air encounter, the pilot does not blindly execute the \ac{ra} maneuver.  Instead, he combines the \ac{ra} with other sources of information (i.e., instrument panel, visual observation) to judge his best course of action.  In doing this, he quantifies the quality of a course of action in terms of a utility function, or degree of happiness, defined over possible results of that course of action. That utility function does not only involve proximity to the other aircraft in the encounter, but also involves how drastic a maneuver the pilot makes. For example, if the pilot believes that a collision is unlikely based on his observations, he may opt to ignore the alarm and continue on his current course, thereby avoiding any loss of utility incurred by maneuvering. This is why a pilot will rationally decide to ignore alarms with a high probability of being false.

When designing TCAS, a high false alarm rate need not be bad in and of itself. Rather what is bad is a high false alarm rate combined with a pilot's utility function to result in pilot behavior which is undesirable at the system level. This more nuanced perspective allows far more powerful and flexible design of alarm systems than simply worrying about the false positive rate.
Here, this perspective is elaborated. We use a semi net-form game for predicting the behavior of a system comprising automation and humans who are motivated by utility functions and anticipation of one another's behavior.

Recall the definition of a semi net-form game via a quintuple ($G,X,u,R,\pi$) in Definition~\ref{def:snfg}.  We begin by specifying each component of this quintuple.  To increase readability, sometimes we will use (and mix) the equivalent notation $Z=X_Z$, $z=x_Z$, and $z'=x'_Z$ for a node $Z$ throughout the \ac{tcas} modeling.

\subsection{Directed Acyclic Graph $G$}
\label{sec:dagG}

The directed acyclic graph $G$ for a 2-aircraft encounter is shown in Figure~\ref{fig:BnetDiagram}.  At any time $t$, the true system state of the mid-air encounter is represented by the world state $S$, which includes the states of all aircraft.  Since the pilots (the players in this model) and \ac{tcas} hardware are not able to observe the world state perfectly, a layer of nodes is introduced to model observational noise and incomplete information.  The variable $W_i$ represents pilot $i$'s observation of the world state, while $W_{TCAS_i}$ represents the observations of \ac{tcas} $i$'s sensors.  A simplified model of the current \ac{tcas} logic is then applied to $W_{TCAS_i}$ to emulate an \ac{ra} $T_i$.  Each pilot uses his own observations $W_i$ and $T_i$ to choose an aircraft maneuver command $A_i$.  Finally, we produce the outcome $H$ by simulating the aircraft states forward in time using a model of aircraft kinematics, and calculate the social welfare $F$.  We will describe the details of these variables in the following sections.

\begin{figure}
\begin{center}
\includegraphics[width=\textwidth,trim = 0 120 0 120]{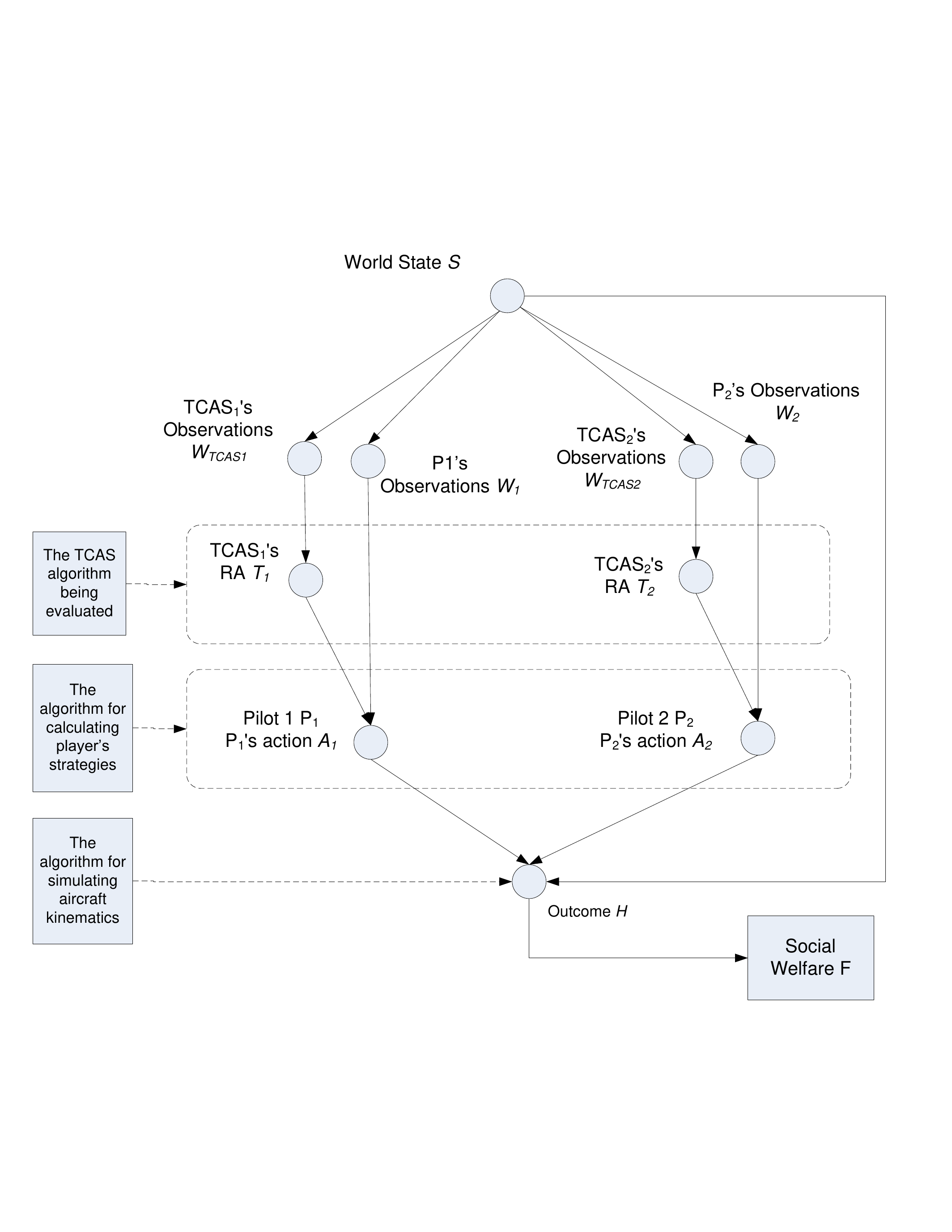}
\caption{Bayes net diagram of a 2-aircraft mid-air encounter:  Each pilot chooses a vertical rate to maximize his expected utility based on his \ac{tcas} alert and a noisy partial observation of the world state.}
\label{fig:BnetDiagram}       
\end{center}
\end{figure}

\subsection{Variable Spaces $X$}
\label{sec:SetX}

\subsubsection{Space of World State $S$}
\label{sec:nodeSx}

The world state $S$ contains all the states used to define the mid-air encounter environment.  It includes 10 states per aircraft to represent kinematics and pilot commands (see Table~\ref{tab:acKinematics}) and 2 states per aircraft to indicate \ac{tcas} intent.  Recall that \ac{tcas} has coordination functionality, where it broadcasts its intents to other aircraft to avoid issuing conflicting \acp{ra}.  The TCAS intent variables are used to remember whether an aircraft has previously issued an \ac{ra}, and if so, what was the sense (direction).

\begin{table}
\caption{A description of aircraft kinematic states and pilot inputs.}
\label{tab:acKinematics}       
\begin{tabular}{p{2.0cm}p{2.0cm}p{7cm}}
\hline\noalign{\smallskip}
Variable & Units & Description  \\
\noalign{\smallskip}\svhline\noalign{\smallskip}
$x$ & ft & Aircraft position in x direction\\
$y$ & ft & Aircraft position in y direction\\
$z$ & ft & Aircraft position in z direction\\
$\theta$ & rad & Heading angle\\
$\dot{\theta}$ & rad/s & Heading angle rate\\
$\dot{z}$ & ft/s & Aircraft vertical speed\\
$f$ & ft/s & Aircraft forward speed\\
$\phi_c$ & rad & Commanded aircraft roll angle\\
$\dot{z}_c$ & ft/s & Commanded aircraft vertical speed\\
$f_c$ & ft/s & Commanded aircraft forward speed\\
\noalign{\smallskip}\hline\noalign{\smallskip}
\end{tabular}
\end{table}

\subsubsection{Space of TCAS Observation $W_{TCAS_i}$}
\label{sec:nodeWtcasX}

Being a physical system, \ac{tcas} does not have full and direct access to the world state.  Rather, it must rely on noisy partial observations of the world to make its decisions.  $W_{TCAS_i}$ captures these observational imperfections, modeling \ac{tcas} sensor noise and limitations.  Note that each aircraft has its own \ac{tcas} hardware and makes its own observations of the world.  Consequently, observations are made from a particular aircraft's perspective.  To be precise, we denote $W_{TCAS_i}$ to represent the observations that TCAS $i$ makes of the world state, where TCAS $i$ is the \ac{tcas} system on board aircraft $i$.  Table~\ref{tab:Wtcas} describes each variable in $W_{TCAS_i}$.  Variables are real-valued (or positively real-valued where the negative values do not have physical meaning).

\begin{table}
\caption{A description of $W_{TCAS_i}$ variables.}
\label{tab:Wtcas}       
\begin{tabular}{p{2.0cm}p{2.0cm}p{7cm}}
\hline\noalign{\smallskip}
Variable & Unit & Description  \\
\noalign{\smallskip}\svhline\noalign{\smallskip}
$r_h$ & ft & Horizontal range between own and intruding aircraft\\
$\dot{r}_h$ & ft/s & Horizontal range rate\\
$\dot{h}$ & ft/s & Relative vertical rate between own and intruding aircraft\\
$h$ & ft & Relative altitude between own and intruding aircraft\\
$h_i$ & ft & Own aircraft's altitude\\
\noalign{\smallskip}\hline\noalign{\smallskip}
\end{tabular}
\end{table}

\subsubsection{Space of \ac{tcas} \ac{ra} $T_i$}
\label{sec:nodeTx}

A simplified version of \ac{tcas}, called mini TCAS, is implemented based on \protect{\cite{KochenderferATC360}} with minor modifications (we will discuss the differences in Section~\ref{sec:nodeTpi}).  Mini TCAS issues an \ac{ra} $T_i$ based on $W_{TCAS_i}$ input, emulating the \ac{tcas} logic.  The variable $T_i$ represents the recommended target vertical rate issued to pilot $i$.  We model $T_i$ as able to take on one of 6 possible values: no \ac{ra} issued, descend at 42 ft/s, descend at 25 ft/s, level-off, climb at 25 ft/s, or climb at 42 ft/s. i.e., $T_i \in \{\emptyset, -42, -25, 0, 25, 42\}$, where $T_i = \emptyset$ indicates no \ac{ra} issued, otherwise $T_i$ is specified in ft/s.

\subsubsection{Space of Pilot's Observation $W_i$}
\label{sec:nodeWx}

Aside from the \ac{ra}, pilots have other sources of information about the world, such as those coming from instruments and visual observations.  All paths of information are considered when the pilot decides his best course of action.  Unfortunately, instruments and visuals also provide noisy partial observations of the world state.

Properly speaking there are many intricacies that should be considered in the pilot observation model.  First, the model should reflect the type and amount of information available via the cockpit instruments.  Secondly, the model should reasonably approximate the visual observation characteristics and its limitations, such as field of view and geometry.  For example, visual accuracy should decrease as distance increases, and moreover, visual observations of the intruding aircraft cannot be acquired altogether if the intruding aircraft is situated behind own aircraft.  Lastly, the model should consider a pilot attention model, since pilots may miss detecting an intruding aircraft if they are preoccupied with other tasks.  Attention and situational awareness are large topics of research in psychology and human factors especially under the context of pilots and military personnel \protect{\cite{Endsley88, Endsley89,Taylor90, Watts04}}.

As a first step, we do not consider the above subtleties, and begin with a very crude model for the pilot's observations. In particular, we model the pilot's observation $W_i$ as being a corrupted version of $S$.

\subsubsection{Space of Pilot's Move $A_i$ and Outcome $H$}
\label{sec:nodeAiHx}

At his decision point, pilot $i$ observes his parent nodes and takes an action (i.e, sets the value of node $A_i$).  The variable $A_i$ is the target vertical rate for aircraft $i$ represented by a real-valued number between -50 and 50 ft/s.  We take the outcome $H$ as being in the same space as $S$.  We will later see how this facilitates the simulation of the encounter.

\subsection{Utility Function $u$}
\label{sec:UtilFcn}

The pilot's utility function summarizes in a real number the degree of happiness for a given joint outcome.  It is a simple parameterization of the player that characterizes him and summarizes his preferences.  Players act to maximize his expected utility.  In this modeling, we assume all pilots have the same utility function.

Properly speaking, the utility function should be learned from data to be as realistic as possible.  However, since the task of learning parameters from data is a significant research topic of its own, it is being pursued as a separate effort.  For now, the utility functions are crafted using intuition gained from pilot interviews.  The authors found that pilots considered primarily 3 priorities when deciding how to react to a \ac{tcas} \ac{ra}.

\begin{enumerate}
\item \textbf{Avoid collision.}  Pilots will do all that is necessary to avoid a collision, thus this has the highest priority.  Since the fear of a collision increases as the aircraft get closer, a representative metric for measuring this impetus for collision avoidance is the minimum approach distance between aircraft $d_{min}$ (i.e., the smallest separation distance between two aircraft over the entire encounter).
\item \textbf{Course deviation.}  There are many reasons why a pilot does not want to deviate from his current course.  For example, deviations are often associated with longer flight times, higher fuel consumption, and increased flying effort.  The notion is that if a collision is deemed unlikely (i.e., there's a high chance of \ac{tcas} being a false positive), then the pilot will be inclined to stay on his current course.  We reflect this inclination by penalizing the difference between the current vertical speed and the vertical speed in consideration.
\item \textbf{Obeying \ac{tcas}.} Pilots have indicated that when they feel uncertain that they will be inclined to follow \ac{tcas}.  In other words, given all else equal, pilots have a natural tendency to follow \acp{ra}.  This may be attributed to their training, their inclination to follow orders, or even blind trust in the system.  We summarize this tendency into a metric by penalizing moves that deviate from the issued \ac{ra}.
\end{enumerate}

In summary, utility function is chosen to be of the following form:
\begin{eqnarray}
u_i = \alpha_1 \log(\delta + d_{min}) - \alpha_2 |\dot{z} - a_i| - \alpha_3 |T_i - a_i| \nonumber
\end{eqnarray}
\noindent where $\alpha_1$, $\alpha_2$, and $\alpha_3$ are real positive constant weights, $\delta$ is a small positive constant, $a_i$ is the pilot's action, $d_{min}$ is the minimum approach distance between the aircraft, and $\dot{z}$ is the aircraft's current vertical speed.  The weights reflect how the pilot trades off the three competing objectives.  The weight $\alpha_1$ is largest, since avoiding collision is highest priority; $\alpha_2$ is the second largest, followed by $\alpha_3$ with the smallest weight.  The log function in the first term is used to capture the fact that the rate of utility increase/decrease is much faster when the aircraft are close together than when they are far apart.

\subsection{Partition $R$}
\label{sec:PartitionR}

We partition the variables in the net into chance and decision nodes as follows:  The nodes in the set $\{S, W_{TCAS_1}, W_{TCAS_2}, T_1, T_2, W_1, W_2,H\}$ are chance nodes, and the nodes $A_1$ and $A_2$ are decision nodes.  Moreover, player 1 sets the value of the node $A_1$, and player 2 sets the value of the node $A_2$.

\subsection{Set of Conditional Probabilities $\pi$}
\label{sec:SetPi}

In this section, we describe the conditional probability distribution associated with each chance node.  Note the use of stochastic terminology such as ``sample" and ``conditional probability distribution" for both stochastic and deterministic nodes.  This is in light that we may view a deterministic node as stochastic with all its probability mass on its deterministic result.

\subsubsection{CPD of the World State $S$}
\label{sec:nodeSpi}

At the beginning of an encounter, the world state is initialized using the encounter generator (to be discussed in Section~\ref{sec:encGenerator}).  Otherwise, the outcome $H$ at time $t-\Delta t$ becomes the world state $S$ at time $t$.

\subsubsection{CPD of TCAS Observation $W_{TCAS_i}$}
\label{sec:nodeWtcasPi}

To calculate $W_{TCAS_i}$, the exact versions of the variables in Table~\ref{tab:Wtcas} are first calculated from the world state $S$ using the following equations:
\begin{eqnarray}
r_h &=& \sqrt{(x_j-x_i)^2 + (y_j-y_i)^2} \nonumber \\
\dot{r}_h &=& \frac{1}{r_h} \cdot ((x_j-x_i)(f_j \cos\theta_j-f_i \cos\theta_1)+(y_j-y_i)(f_j \sin\theta_j- f_i\sin\theta_i)) \nonumber \\
\dot{h} \; &=& \dot{z}_j - \dot{z}_i \nonumber \\
h \; &=& z_j - z_i \nonumber \\
h_i &=& z_i \nonumber
\end{eqnarray}
\noindent where the subscripts $i$ and $j$ indicate own and intruding aircraft, respectively.  We then add zero-mean Gaussian noise\footnote{We assume independent noise for each variable $r_h,\dot{r}_h,\dot{h},h,h_i$ with standard deviations of $100,50,4,10,10$, respectively.  These variables are described in Table~\ref{tab:Wtcas}.} to the variables to emulate sensor noise.

\subsubsection{CPD of \ac{tcas} \ac{ra} $T_i$}
\label{sec:nodeTpi}
Recall from Section~\ref{sec:nodeTx} that we use mini TCAS to emulate the full \ac{tcas} logic.  The major assumptions of mini TCAS are:

\begin{enumerate}
\item All aircraft are \ac{tcas}-equipped and coordinate \acp{ra}.
\item Actual horizontal range rate is used\footnote{In \protect{\cite{KochenderferATC360}}, the horizontal range rate is fixed to -500 ft/s.}.
\item No tracking or encounter monitoring over time is performed.  Hence, mini TCAS is a memory-less system.
\item No \ac{tcas} strengthenings or reversals (updates or revisions to the initial \ac{tcas}).
\item The tau-rising test and horizontal miss distance test are not performed \protect{\cite{KochenderferATC360}}.
\end{enumerate}

The implementation of mini TCAS in this work follows closely that described in \protect{\cite{KochenderferATC360}}.  First, the range, altitude, and altitude separation tests are performed for collision detection.  If no potential collision is detected, no \ac{ra} is issued.  If a potential collision is detected, the algorithm then continues to determine the sense (direction) and strength of the \ac{ra}.  In the sense selection process, the algorithm determines which direction (ascend, descend, or level-off) gives the greatest vertical miss distance.  However, to account for \ac{tcas} coordination, a modification to the algorithm is made.  To avoid issuing conflicting \acp{ra}, intruders' senses (up, level, or down) that appear in received intent messages are first removed from the list of candidate senses for own aircraft.  The direction is chosen from the remaining choices.  Strength selection follows to choose the least disruptive \ac{ra} that still achieves the minimum safety distance.

It is known that pilots react differently to a revised (second) \ac{ra} than the initial one.  Especially in cases of the \acp{ra} contradicting one another, the pilots may experience cognitive dissonance, and even go into a confused mental state.  As a result, to model this phenomenon properly would require a whole new level of pilot modeling, with perhaps separate models for the first and second \acp{ra}. One possible hack is to use the same model for both \acp{ra}.  However, doing so would yield misleading results, since the pilot would experience no ``mental conflict" to go against the previous \ac{ra}, and thus is much more likely to comply to any \ac{ra} change. Alternatively, social welfare $F$ could be hacked to demerit reversals or strenghtenings to \acp{ra}. For now though, we discard any encounters that issue reversals or strenghtenings.

\subsubsection{CPD of Pilot's Observation $W_i$ and Outcome $H$}
\label{sec:nodeWpi}

We model the pilot's observation $W_i$ as being $S$ corrupted with additive zero-mean Gaussian noise\footnote{We assume independent noise for each variable $x,y,z,\theta,\dot{\theta},\dot{z},f$ with standard deviations of $100,100,20,0.05,0,5,10$, respectively.  These variables are described in Table~\ref{tab:acKinematics}.}.  The outcome $H$ is calculated using the world state $S$, the pilot actions $A_1, A_2$, and the aircraft kinematics described in Section~\ref{sec:acKinematics}. 

\subsection{Computing Level-K d-Relaxed Strategies}
\label{sec:computeStrategies}
Using the semi net-form game ($G,X,u,R,\pi$) specified previously, this section describes the application of a modified version of Definition~\ref{def:unnormLwSampling} to calculate player $i$'s strategies.  Table~\ref{tab:LkParameters} specifies the additional parameters needed by the algorithm.  To be realistic, these model parameters should be learned from real data.  However, for now, they are chosen by hand.  For convenience, we define the new variable $WT_{i'}$, where $i'$ is a dummy player index, to be the combination of the nodes $W_{i'}$, $W_{TCAS_{i'}}$, and $T_{i'}$.  
Let $v$ be the node controlled by player $i$.  Then, applying Definition~\ref{def:rejSamp_symbols}, we see that $v=A_i$, $pa(v) = \{W_i,T_i\}$, $S^v=\{H\}$, $S^{-v}=\{S,WT_i,WT_{-i},A_{-i}\}$, and $Y=\{S,W_{TCAS_i},WT_{-i},A_{-i}\}$.

\begin{table}
\caption{Specification of parameters in unnormalized likelihood-weighted level-K d-relaxed strategies (Definition~\ref{def:unnormLwSampling}) for the collision avoidance problem.}
\label{tab:LkParameters}       
\begin{tabular}{p{2.0cm}p{3.5cm}p{5.5cm}}
\hline\noalign{\smallskip}
Parameter & Value & Description  \\
\noalign{\smallskip}\svhline\noalign{\smallskip}
$K$ & $2$ & Player level for all pilots\\
$M$ & $5$ & Number of samples of pilot's own movespace\\
$M'$ & $10$ & Number of samples of the pilot's environment\\
$\lambda(A_i \mid w_i, t_i)$ & Uniform over movespace & Satisficing distribution of player $i$\\
$\Lambda^0(A_i \mid w_i,t_i)$ & Wide Gaussian ($\sigma=20$) about \ac{ra} & Level 0 distribution of player $i$\\
\noalign{\smallskip}\hline\noalign{\smallskip}
\end{tabular}
\end{table}

We proceed following the steps of Definition~\ref{def:unnormLwSampling}.  In step 1, we form a set $\{a'_i(j) : j = 1, \ldots, M\}$ by IID sampling $\lambda(A_i \mid w_i, t_i)$ M times.  Since the space of $A_i$ is continuous, we do not need to worry about removing duplicates.  

The application of step 2 requires a slight modification.
Recall that \ac{tcas} logic is deterministic, causing its probability $\tilde{P}^K(t_i \mid w'_{TCAS_i})$ where $t_i$ is the observed realization of $T_i$, to be either 1 or 0.  This creates a natural filtering effect that zeroes out entire posterior probabilities in the sum according to whether the scenarios cause the observed (evidence) \ac{ra} to occur.  In fact, since the space of the world state $S$ is so large, we found the number of unusable samples to be impractically high.  This rendered the straightforward application of Definition~\ref{def:unnormLwSampling} infeasible.

To help direct the samples toward the relevant subspace, we introduce importance sampling to propose nodes $S$ and $W_{TCAS_i}$ using their values sampled at the top-level $s$ and $w_{TCAS_i}$ respectively.  Note that the player does \emph{not} have access to $w_{TCAS_i}$ or $s$ -- rather the simulator does. We use these variables to form the proposal distribution for approximating the expectation.  More concretely, we form a set of weight-sample pairs $\{(w'(k),x'_{S^{-v}}(k)) : k=1,\ldots M'\}$ by setting $w'_i = w_i$ and $t'_i = t_i$, and instead of sampling from $\tilde{P}^K(X_Y) = \prod_{v' \in Y}\tilde{P}^K(X_{v'} \mid X_{pa(v')})$ as described in step 2 of Definition~\ref{def:unnormLwSampling}, we IID sample from: 
\begin{eqnarray}
Q(X_Y) = \prod_{v' \in Y \setminus \{S,W_{TCAS_i}\}}\tilde{P}^K(X_{v'} \mid X_{pa(v')}) Q(S \mid s) Q(W_{TCAS_i} \mid w_{TCAS_i}) \nonumber
\end{eqnarray}
and adjust the weighting factor accordingly by multiplying $w'(k)$ by: 
\begin{eqnarray}
\frac{\tilde{P}^K(s'(k))}{Q(s'(k) \mid s)} \frac{\tilde{P}^K(w'_{TCAS_i}(k) \mid s'(k))}{Q(w'_{TCAS_i}(k) \mid w_{TCAS_i})} \nonumber
\end{eqnarray}
One can verify that this manipulation does not change the expected value~\protect{\cite{RobertBook}}.  Recall that $S$ is composed of two parts: one that contains the kinematic states of the aircraft, and the other that represents the \ac{tcas} intent messages.  We denote these variables as $S_K$ and $S_I$, respectively, and choose to propose them separately, i.e., $Q(S \mid s) = Q(S_K \mid s_K) Q(S_I \mid s_I)$.  We choose $Q(S_K \mid s_K)$ to be a tight Gaussian distribution\footnote{We assume independent noise for each variable $x,y,z,\theta,\dot{\theta},\dot{z},f$ with standard deviations of $5,5,2,0.01,0,1,5$, respectively.  These variables are described in Table~\ref{tab:acKinematics}.} centered about $s_K$, and choose $Q(S_I \mid s_I)$ to be a delta function about the true value $s_I$ with probability $q$, or one of the following 4 values each with probability $\frac{1}{4}(1-q)$: \\
1. No intent received. \\
2. Intent received with an up sense. \\
3. Intent received with a level-off sense. \\
4. Intent received with a down sense. \\
We choose $Q(W_{TCAS_i} \mid w_{TCAS_i})$ to be a tight Gaussian distribution\footnote{We assume independent noise for each variable $r_h,\dot{r}_h,\dot{h},h,h_i$ with standard deviations of $5,2,2,2,2$, respectively.  These variables are described in Table~\ref{tab:Wtcas}.} centered about $w_{TCAS_i}$.

The trick, as always with importance sampling Monte Carlo, is to choose a proposal distribution that will result in low variance, and that is nonzero at all values where the integrand is nonzero~\protect{\cite{RobertBook}}.  In this case, so long as the proposal distribution over $s'$ has full support, the second condition is met.  So the remaining issue is how much variance there will be. Since $Q(W_{TCAS_i} \mid w_{TCAS_i})$ is a tight Gaussian by the choice of proposal distribution, values of $w'_{TCAS_i}$ will be very close to values of $w_{TCAS_i}$, causing $P(t_i \mid w'_{TCAS_i})$ to be much more likely to equal 1 than 0.  To reduce the variance even further, rather than form $M'$ samples of the distribution, samples of the proposal distribution are generated until $M'$ of them have nonzero posterior.

We continue at step 3.  For each candidate action $a'_i(j)$, we estimate its expected utility by sampling the outcome $H$ from $\tilde{P}^K(H \mid x'_Y(;),a'_i(j),w_i,t_i)$, and computing the estimate $\tilde{u}_i^K$.  The weighting factor compensates for our use of a proposal distribution to sample variables rather than sampling them from their natural distributions.  Finally, in step 4, we choose the move $a'_i(j^*)$ that has the highest expected utility estimate.

\subsection{Encounter Simulation}
\label{sec:sysFlow}

Up until now, we have presented a game which describes a particular instant $t$ in time.  In order to simulate an encounter to any degree of realism, we must consider how this game evolves with time.

\subsubsection{Time Extension of the Bayes Net}
\label{sec:timeExtension}

Note that the timing of decisions\footnote{We are referring to the time at which the player makes his decision, not the amount of time it takes for him to decide.  Recall that level-K reasoning occurs only in the mind of the decision-maker and thus does not require any wall clock time.} is in reality stochastic as well as asynchronous.  However, to consider a stochastic and asynchronous timing model would greatly increase the model's complexity.  For example, the pilot would need to average over the other pilots' decision timing and vice versa.  As a first step, we choose a much simpler timing model and make several simplifying assumptions.  First, each pilot only gets to choose a single move, and he does so when he receives his initial \ac{ra}.  This move is maintained for the remainder of the encounter.  Secondly, each pilot decides his move by playing a simultaneous move game with the other pilots (the game described by ($G,X,u,R,\pi$)).  These assumptions effectively remove the timing stochasticity from the model.

The choice of modeling as a simultaneous move game is an approximation, as it precludes the possibility of the player anticipating the timing of players' moves.  Formally speaking, this would introduce an extra dimension in the level-K thinking, where the player would need to sample not only players' moves, but also the timing of such a move for all time in the past and future.  However, it is noted that since the players are not able to observe each other's move directly (due to delays in pilot and aircraft response), it does not make a difference to him whether it was made in the past or simultaneously. This makes it reasonable to model the game as simultaneous move at the time of decision.  The subtlety here is that the player's thinking should account for when his opponent made his move via imagining what his opponent would have seen at the time of decision.  However, in this case, since our time horizons are short, this is a reasonable approximation.

Figure~\ref{fig:timeExtBnetDiagram} shows a revised Bayes net diagram -- this time showing the extension to multiple time steps.  Quantities indicated by hatching in the figure are passed between time steps.  There are two types of variables to be passed.  First, we have the aircraft states.  Second, recall that \ac{tcas} intents are broadcasted to other aircraft as a coordination mechanism. These intents must also be passed on to influence future \acp{ra}.

\begin{figure}
\begin{center}
\includegraphics[width=\textwidth,trim = 0 80 0 80]{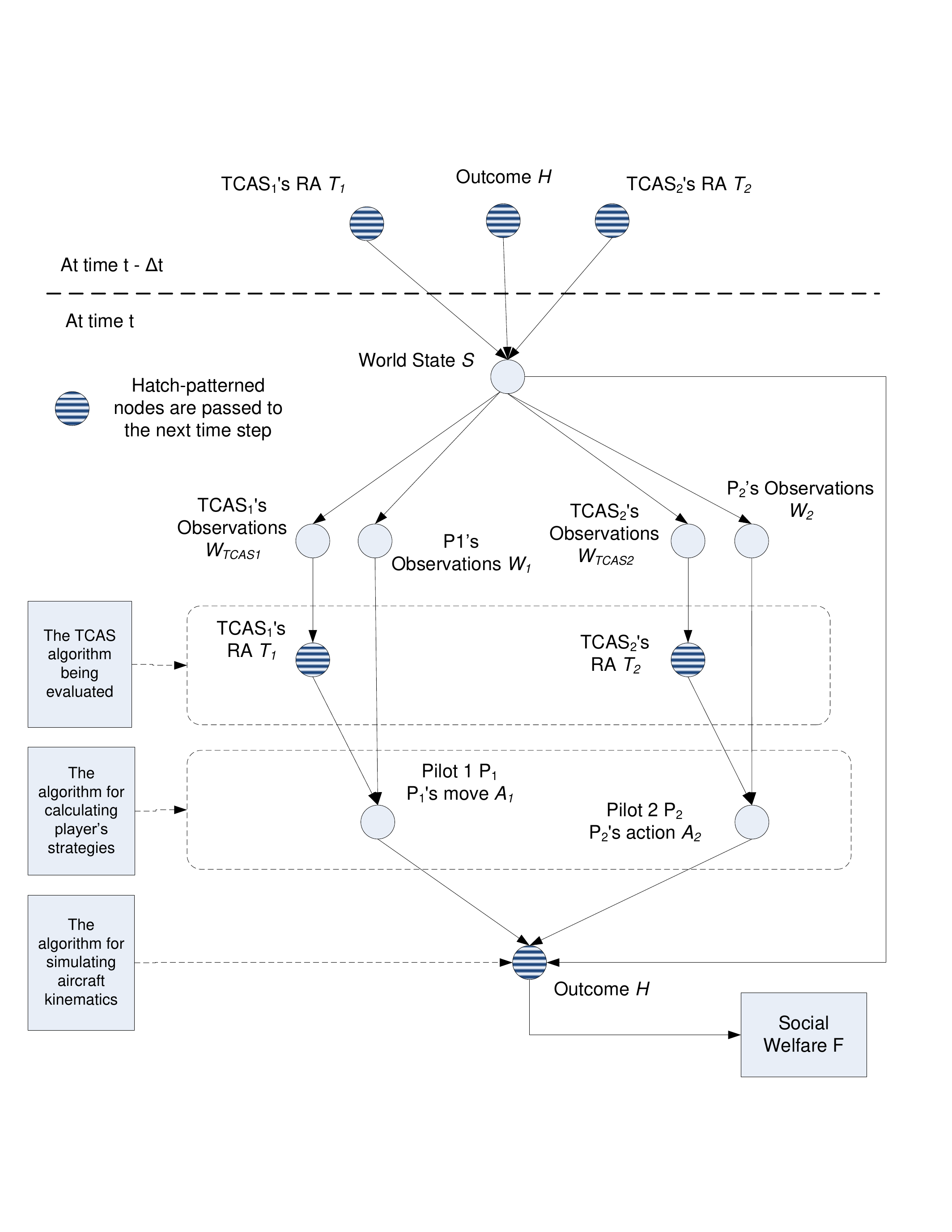}
\caption{Time-extended Bayes net diagram of a 2-aircraft mid-air encounter:  We use a simple timing model that allows each pilot to make a single move at the time he receives his \ac{tcas} alert.}
\label{fig:timeExtBnetDiagram}       
\end{center}
\end{figure}

\subsubsection{Simulation Flow Control}
\label{sec:flowControl}

Using the time-extended Bayes net as the basis for an inner loop, we add flow control to manage the simulation.  Figure~\ref{fig:sysFlowDiagram} shows a flow diagram for the simulation of a single encounter.  An encounter begins by randomly initializing a world state from the encounter generator (to be discussed in Section~\ref{sec:encGenerator}).  From here, the main loop begins.

First, the observational ($W_i$ and $W_{TCAS_i}$) and \ac{tcas} ($T_i$) nodes are sampled.  If a new \ac{ra} is issued, the pilot receiving the new \ac{ra} is allowed to choose a new move via a modified level-K d-relaxed strategy (described in Section~\ref{sec:computeStrategies}).  Otherwise, the pilots maintain their current path.  Note that in our model, a pilot may only make a move when he receives a new \ac{ra}.  Since \ac{tcas} strengthenings and reversals (i.e., updates or revisions to the initial \ac{ra}) are not modeled, this implies that each pilot makes a maximum of one move per encounter.  Given the world state and pilot commands, the aircraft states are simulated forward one time step, and social welfare (to be discussed in Section~\ref{sec:socialWelfare}) is calculated.  If a \acl{nmac} (NMAC) is detected (defined as having two aircraft separated less than 100 ft vertically and 500 ft horizontally) then the encounter ends in collision and a social welfare value of zero is assigned.  If an \ac{nmac} did not occur, successful resolution conditions (all aircraft have successfully passed each other) are checked.  On successful resolution, the encounter ends without collision and the minimum approach distance $d_{min}$ is returned.  If neither of the end conditions are met, the encounter continues at the top of the loop by sampling observational and \ac{tcas} nodes at the following time step.  

\begin{figure}
\begin{center}
\includegraphics[width=\textwidth,trim = 0 50 0 50]{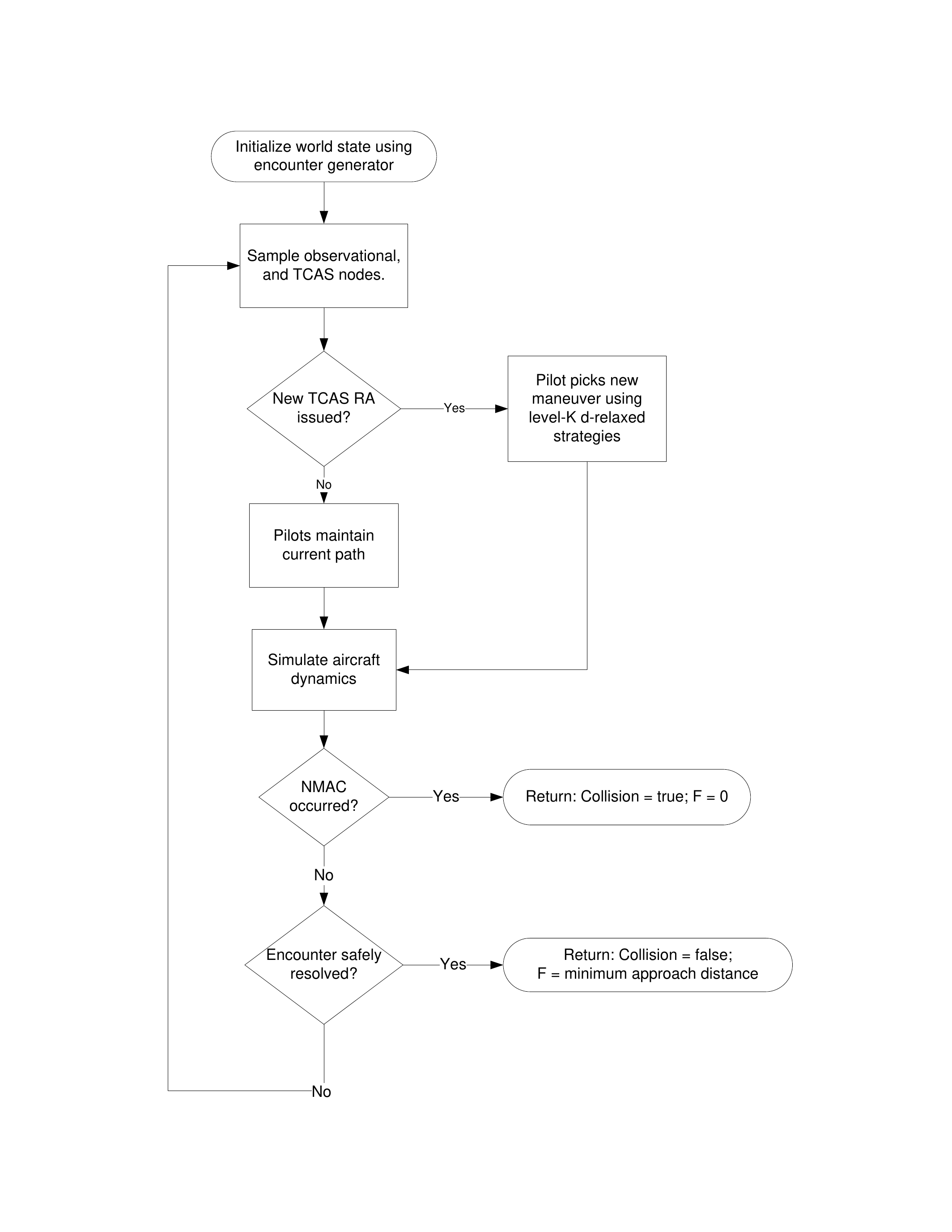}
\caption{Flow diagram of the encounter simulation process: We initialize the encounter using an encounter generator, and simulate forward in time using pilot commands and aircraft kinematics.  The encounter ends when the aircraft have either collided or have successfully passed each other.}
\label{fig:sysFlowDiagram}       
\end{center}
\end{figure}

\subsubsection{Encounter Generation}
\label{sec:encGenerator}

The purpose of the encounter generator is to randomly initialize the world states in a manner that is genuinely representative of reality.  For example, the encounters generated should be of realistic encounter geometries and configurations.  One way to approach this would be to use real data, and moreover, devise a method to generate fictitious encounters based on learning from real ones, such as that described in \protect{\cite{KochenderferATC344,KochenderferATC345}}.  For now, the random geometric initialization described in \protect{\cite{KochenderferItaly11}} Section~6.1 is used\footnote{The one modification is that $t_{target}$ (the initial time to collision between aircraft) is generated randomly from a uniform distribution between 40 s and 60 s rather than fixed at 40 s.}.

\subsubsection{Aircraft Kinematics Model}
\label{sec:acKinematics}

Since aircraft kinematic simulation is performed at the innermost step, its implementation has an utmost impact on the overall system performance.  To address computational considerations, a simplified aircraft kinematics model is used in place of full aircraft dynamics.  We justify these first-order kinematics in 2 ways:  First, we note that high-frequency modes are unlikely to have a high impact at the time scales ($\sim 1$ min.) that we are dealing with in this modeling.  Secondly, modern flight control systems operating on most (especially commercial) aircraft provide a fair amount of damping of high-frequency modes as well as provide a high degree of abstraction.  We make the following assumptions in our model:
\begin{enumerate}
\item \textbf{Only kinematics are modeled.} Aerodynamics are not modeled.  The assumption is that modern flight control systems abstract this from the pilot.
\item \textbf{No wind.} Wind is not considered in this model.
\item \textbf{No sideslip.} This model assumes that the aircraft velocity vector is always fully-aligned with its heading.
\item \textbf{Pitch angle is abstracted.} Pitch angle is ignored.  Instead, the pilot directly controls vertical rate.
\item \textbf{Roll angle is abstracted.} Roll angle is ignored.  Instead, the pilot directly controls heading rate.
\end{enumerate}

Figure~\ref{fig:acKinFunctional} shows the functional diagram of the kinematics model.  The input commands are first applied as inputs to first-order linear filters to update $\dot{\theta}$, $\dot{z}$, and $f$, these quantities are then used in the kinematic calculations to update the position and heading of the aircraft at the next time step.  Intuitively, the filters provide the appropriate time response (transient) characteristics for the system, while the kinematic calculations approximate the effects of the input commands on the aircraft's position and heading. 

\begin{figure}
\begin{center}
\includegraphics[width=\textwidth,trim = 30 60 0 60]{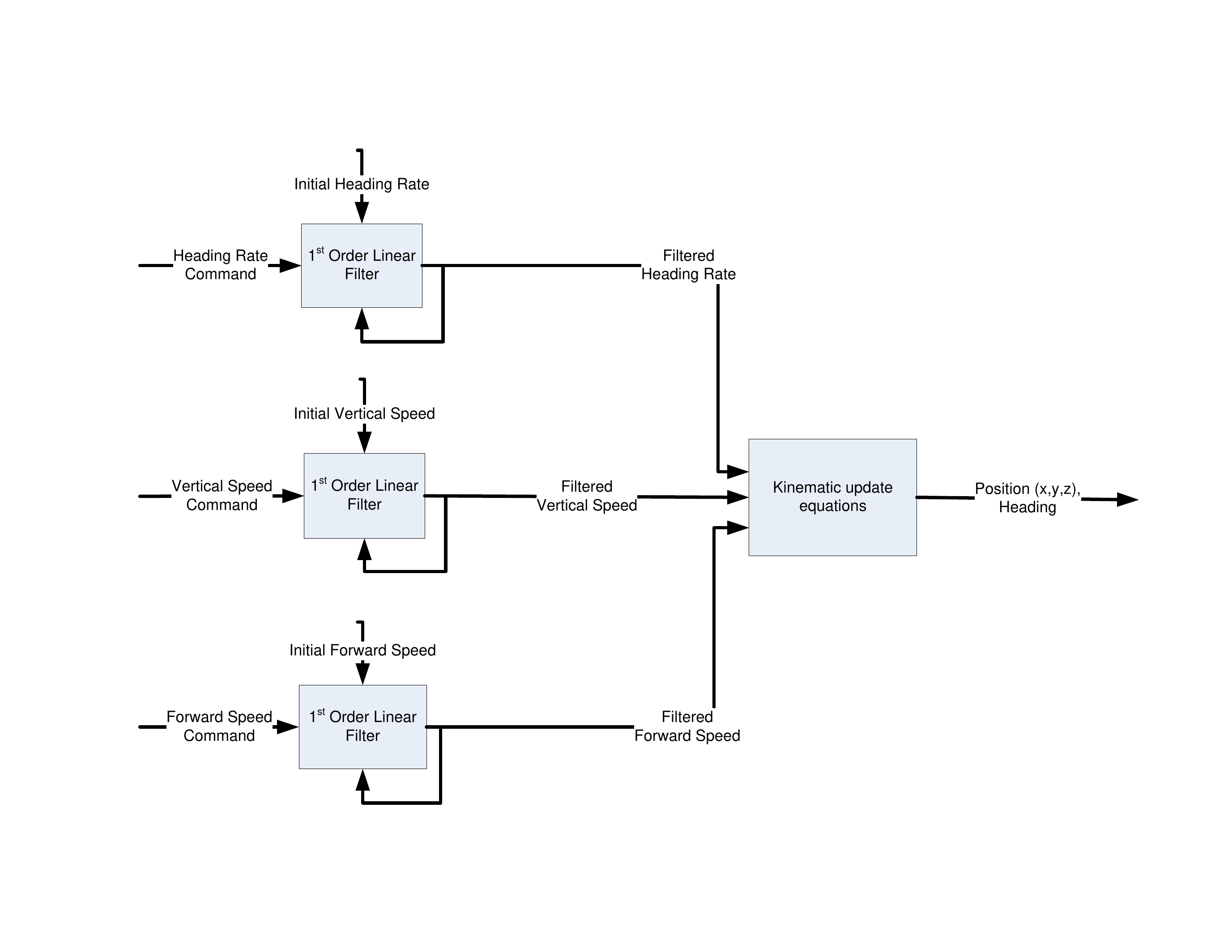}
\caption{Aircraft kinematics model functional diagram: Pilot commands are passed to filters to model aircraft transient response to first order.  Then aircraft kinematic equations based on forward Euler integration are applied.}
\label{fig:acKinFunctional}       
\end{center}
\end{figure}

The kinematic update equations, based on forward Euler integration method, are given by:
\begin{eqnarray}
\theta_{t+\Delta t} &=& \theta_t + \Delta t \cdot \dot{\theta}_t \nonumber \\
x_{t+\Delta t} &=& x_t + \Delta t \cdot f_t \cdot \cos\theta_t \nonumber \\
y_{t+\Delta t} &=& y_t  + \Delta t \cdot f_t \cdot \sin\theta_t \nonumber \\
z_{t+\Delta t} &=& z_t  + \Delta t \cdot \dot{z}_t \nonumber
\end{eqnarray}

Recall that a first-order filter requires two parameters: an initial value and a time constant.  We set the filter's initial value to the pilot commands at the start of the encounter, thus starting the filter at steady-state.  The filter time constants are chosen by hand (using the best judgment of the designers) to approximate the behavior of mid-size commercial jets.  Refinement of these parameters is the subject of future work.

\subsubsection{Modeling Details Regarding the Pilot's Move $A_i$}
\label{sec:nodeAi}

Recall that a pilot only gets to make a move when he receives a new \ac{ra}.  In fact, since strenghtenings and reversals are not modeled, the pilot will begin the scenario with a vertical speed, and get at most one chance to change it.  At his decision point, the pilots engage in a simultaneous move game (described in Section~\ref{sec:computeStrategies}) to choose an aircraft escape maneuver.  To model pilot reaction time, a 5-second delay is inserted between the time the player chooses his move, and when the aircraft maneuver is actually performed.

\subsection{Social Welfare $F$}
\label{sec:socialWelfare}

Social welfare function is a function specified a-priori that maps an instantiation of the Bayes net variables to a real number $F$.  As a player's degree of happiness is summarized by his utility, social welfare is used to quantify the degree of happiness for the system as a whole.  Consequently, this is the system-level metric that the system designer or operator seeks to maximize.  As there are no restrictions on how to set the social utility function, it is up to the system designer to decide the system objectives.  In practice, regulatory bodies, such as Federal Aviation Administration, or International Civil Aviation Organization, will likely be interested in defining the social welfare function in terms of a balance of safety and performance metrics.  For now, social welfare is chosen to be the minimum approach distance $d_{min}$.  In other words, the system is interested in aircraft separation.

\subsection{Example Encounter}
\label{sec:exampleEnc}

To look at average behavior, one would execute a large number of encounters to collect statistics on $F$.  To gain a deeper understanding of encounters, however, we examine encounters individually.  Figure~\ref{fig:vertTraj} shows 10 samples of the outcome distribution for an example encounter.  Obviously, only a single outcome occurs in reality, but the trajectory spreads provide an insightful visualization of the distribution of outcomes.  In this example, we can see (by visual inspection) that a mid-air collision is unlikely to occur in this encounter.  Furthermore, we see that probabilistic predictions by semi net-form game modeling provide a much more informative picture than the deterministic predicted trajectory that the \ac{tcas} model assumes (shown by the thicker trajectory).

\begin{figure}
\begin{center}
\includegraphics[width=\textwidth]{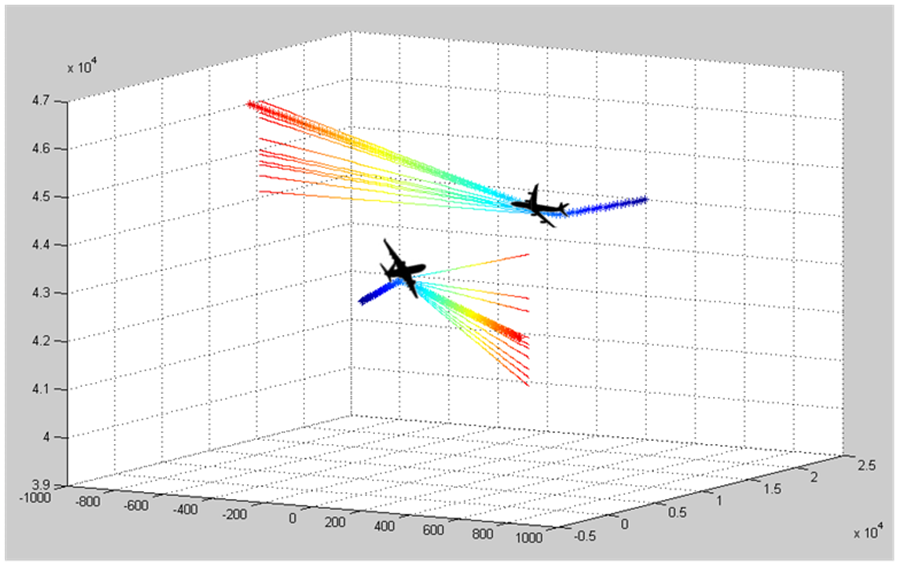}
\caption{Predicted trajectories sampled from the outcome distribution of an example encounter: Each aircraft proceeds on a straight-line trajectory until the pilot receives an \ac{ra}.  At that point, the pilot uses level-K d-relaxed strategies to decide what vertical rate to execute. The resultant trajectories from 10 samples of the vertical rate are shown. The trajectory assumed by \ac{tcas} is shown as the thicker trajectory.}
\label{fig:vertTraj}       
\end{center}
\end{figure}

\subsection{Sensitivity Analysis}
\label{sec:tradeAnalysis}

Because of its sampling nature, level-K relaxed strategy and its variants are all well-suited for use with Monte Carlo techniques. In particular, such techniques can be used to assess performance of the overall system by measuring social welfare $F$ (as defined in Section~\ref{sec:socialWelfare}).  Observing how $F$ changes while varying parameters of the system can provide invaluable insights about a system.  To demonstrate the power of this capability, parameter studies were performed on the mid-air encounter model, and sample results are shown in Figures~\ref{fig:ObsVsG}-\ref{fig:LKBoMvsG}.  In each case, we observe expected social welfare while selected parameters are varied.  Each datapoint represents the average of 1800 encounters.

In Figure~\ref{fig:ObsVsG}, the parameters $M_w$ and $M_{W_{TCAS}}$, which are multiples on the noise standard deviations of $W$ and $W_{TCAS}$ respectively, are plotted versus social welfare $F$.  It can be seen that as the pilot and \ac{tcas} system's observations get noisier (e.g. due to fog or faulty sensors), social welfare decreases.  This agrees with our intuition.  A noteworthy observation is that social welfare decreases faster with $M_w$ (i.e., when the pilot has a poor visual observation of his environment) than with $M_{W_{TCAS}}$ (i.e., noisy TCAS sensors). This would be especially relevant to, for example, a funder who is allocating resources for developing less noisy \ac{tcas} sensors versus more advanced panel displays for better situational awareness.

\begin{figure}
\begin{center}
\includegraphics[width=\textwidth]{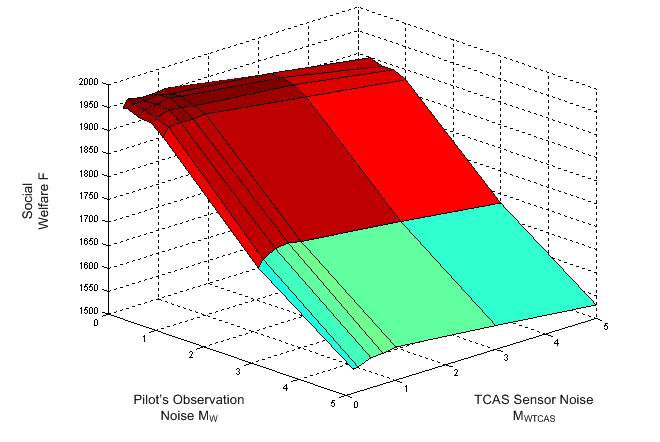}
\caption{Impacts of observational noise on social welfare:  Social welfare is plotted against multiples on the noise standard deviations of $W$ and $W_{TCAS}$.  We observe that social welfare decreases much faster with increase in $M_{W}$ than with increase in $M_{W_{TCAS}}$.  This means that according to our model, pilots receive more information from their general observations of the world state than from the \ac{tcas} \ac{ra}.}
\label{fig:ObsVsG}       
\end{center}
\end{figure} 

Figure~\ref{fig:TCASvsG} shows the dependence of social welfare on selected \ac{tcas} internal logic parameters DMOD and ZTHR \protect{\cite{KochenderferATC360}}.  These parameters are primarily used to define the size of safety buffers around the aircraft, and thus it makes intuitive sense to observe an increase in $F$ (in the manner that we've defined it) as these parameters are increased.  Semi net-form game modeling gives full quantitative predictions in terms of a social welfare metric.

\begin{figure}
\begin{center}
\includegraphics[width=\textwidth]{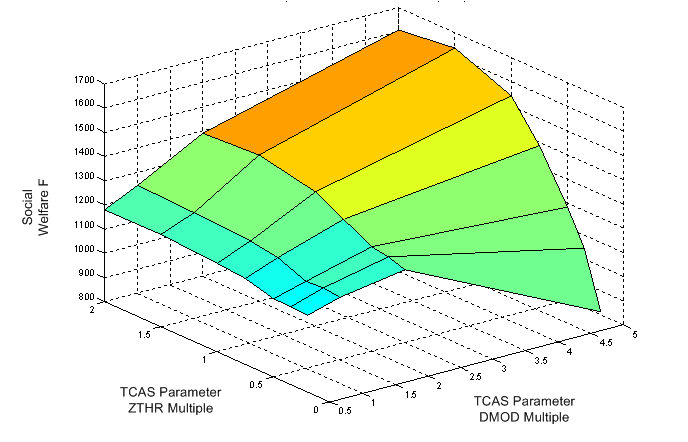}
\caption{Impacts of \ac{tcas} parameters DMOD and ZTHR on social welfare: We observe that social welfare increases as DMOD and ZTHR are increased. This agrees with our intuition since these parameters are used to define the size of safety buffers around the aircraft.}
\label{fig:TCASvsG}       
\end{center}
\end{figure}

Figure~\ref{fig:UtilWeightsVsG} plots player utility weights versus social welfare.  In general, the results agree with intuition that higher $\alpha_1$ (stronger desire to avoid collision) and lower $\alpha_2$ (weaker desire to stay on course) lead to higher social welfare.  These results may be useful in quantifying the potential benefits of training programs, regulations, incentives, and other pilot behavior-shaping efforts.  

\begin{figure}
\begin{center}
\includegraphics[width=\textwidth]{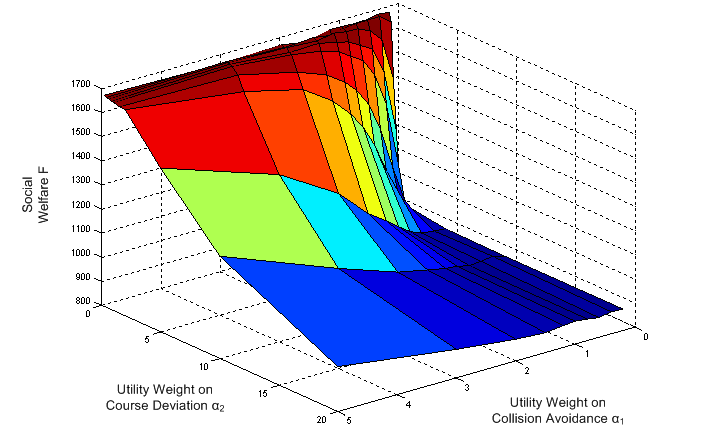}
\caption{Impacts of player utility weights (see Section~\ref{sec:UtilFcn}) on social welfare: We observe that higher $\alpha_1$ (more weight to avoiding collision) and lower $\alpha_2$ (less weight to maintaining current course) leads to higher social welfare.}
\label{fig:UtilWeightsVsG}       
\end{center}
\end{figure}

Figure~\ref{fig:LKBoMvsG} plots model parameters $M$ and $M'$ versus $F$.  Recall from our discussion in Section~\ref{sec:LKRelaxed} that these parameters can be interpreted as a measure of the pilot's rationality.  As such, we point out that these parameters are not ones that can be controlled, but rather ones that should be set as closely as possible to reflect reality.  One way to estimate the ``true" $M$ and $M'$ would be to learn them from real data.  (Learning model parameters is the subject of a parallel research project.)  A plot like Figure~\ref{fig:LKBoMvsG} allows a quick assessment of the sensitivity of $F$ to $M$ and $M'$.

\begin{figure}
\begin{center}
\includegraphics[width=\textwidth]{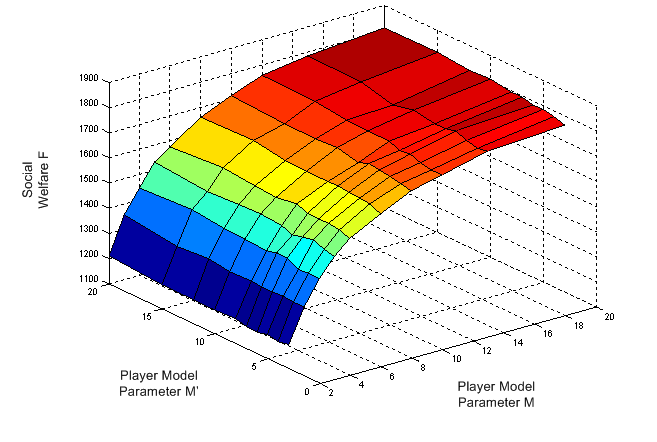}
\caption{Impacts of pilot model parameters $M$ and $M'$ (see Definition~\ref{def:unnormLwSampling}) on social welfare: We observe that as these parameters are increased, there is an increase in social welfare.  This agrees with our interpretation of $M$ and $M'$ as measures of rationality.}
\label{fig:LKBoMvsG}       
\end{center}
\end{figure}

\subsection{Potential Benefits of a Horizontal \ac{ra} System}
\label{sec:horizRA}
Recall that due to high noise in \ac{tcas}' horizontal sensors, the current \ac{tcas} system issues only vertical \acp{ra}.  In this section, we consider the potential benefits of a horizontal \ac{ra} system.  The goal of this work is not to propose a horizontal \ac{tcas} system design, but to demonstrate how semi net-form games can be used to evaluate new technologies.

In order to accomplish this, we make a few assumptions.  Without loss of generality, we refer to the first aircraft to issue an \ac{ra} as aircraft 1, and the second aircraft to issue an \ac{ra} as aircraft 2.  First, we notice that the variable $W_{TCAS_i}$ does not contain relative heading information, which is crucial to properly discriminating between various horizontal geometric configurations.  In \protect{\cite{KochenderferATC371}}, Kochenderfer et al. demonstrated that it is possible to track existing variables (range, range rate, bearing to intruder, etc.) over time using an unscented Kalman filter to estimate relative heading and velocity of two aircraft.  Furthermore, estimates of the steady-state tracking variances for these horizontal variables were provided.  For simplicity, this work does not attempt to reproduce these results, but instead simply assumes that these variables exist and are available to the system.

Secondly, until now the pilots have been restricted to making only vertical maneuvers.  This restriction is now removed, allowing pilots to choose moves that have both horizontal and vertical components.  However, we continue to assume enroute aircraft, and thus aircraft heading rates are initialized to zero at the start of the encounter.  Finally, we assume that the horizontal \ac{ra} system is an augmentation to the existing \ac{tcas} system rather than a replacement.  As a result, we first choose the vertical component using mini TCAS as was done previously, then select the horizontal \ac{ra} component using a separate process.

As a first step, we consider a reduced problem where we optimize the horizontal \ac{ra} for aircraft 2 only; aircraft 1 is always issued a maintain heading horizontal \ac{ra}.  (Considering the full problem would require backward induction, which we do not tackle at this time.)  For the game theoretic reasoning to be consistent, we make the assumption that the \ac{ra} issuing order is known to not only the \ac{tcas} systems, but also the pilots.  Presumably, the pilots would receive this order information via their intrument displays.  To optimize the \ac{ra} horizontal component for aircraft 2, we perform an exhaustive search over each of the five candidate horizontal \acp{ra} (hard left, moderate left, maintain heading, moderate right, and hard right) to determine its expected social welfare.  The horizontal \ac{ra} with the highest expected social welfare is selected and issued to the pilot.  To compute expected social welfare, we simulate a number of counterfactual scenarios of the remainder of the encounter, and then average over them.

To evaluate its performance, we compare the method described above (using exhaustive search) to a system that issues a `maintain heading' \ac{ra} to both aircraft.  Figure~\ref{fig:horizG} shows the distribution of social welfare for each system.  Not only does the exhaustive search method show a higher expected value of social welfare, it also displays an overall distribution shift, which is highly desirable.  By considering the full shape of the distribution rather than just its expected value, we gain much more insight into the behavior of the underlying system.

\begin{figure}
\begin{center}
\includegraphics[width=\textwidth,trim = 120 120 120 120]{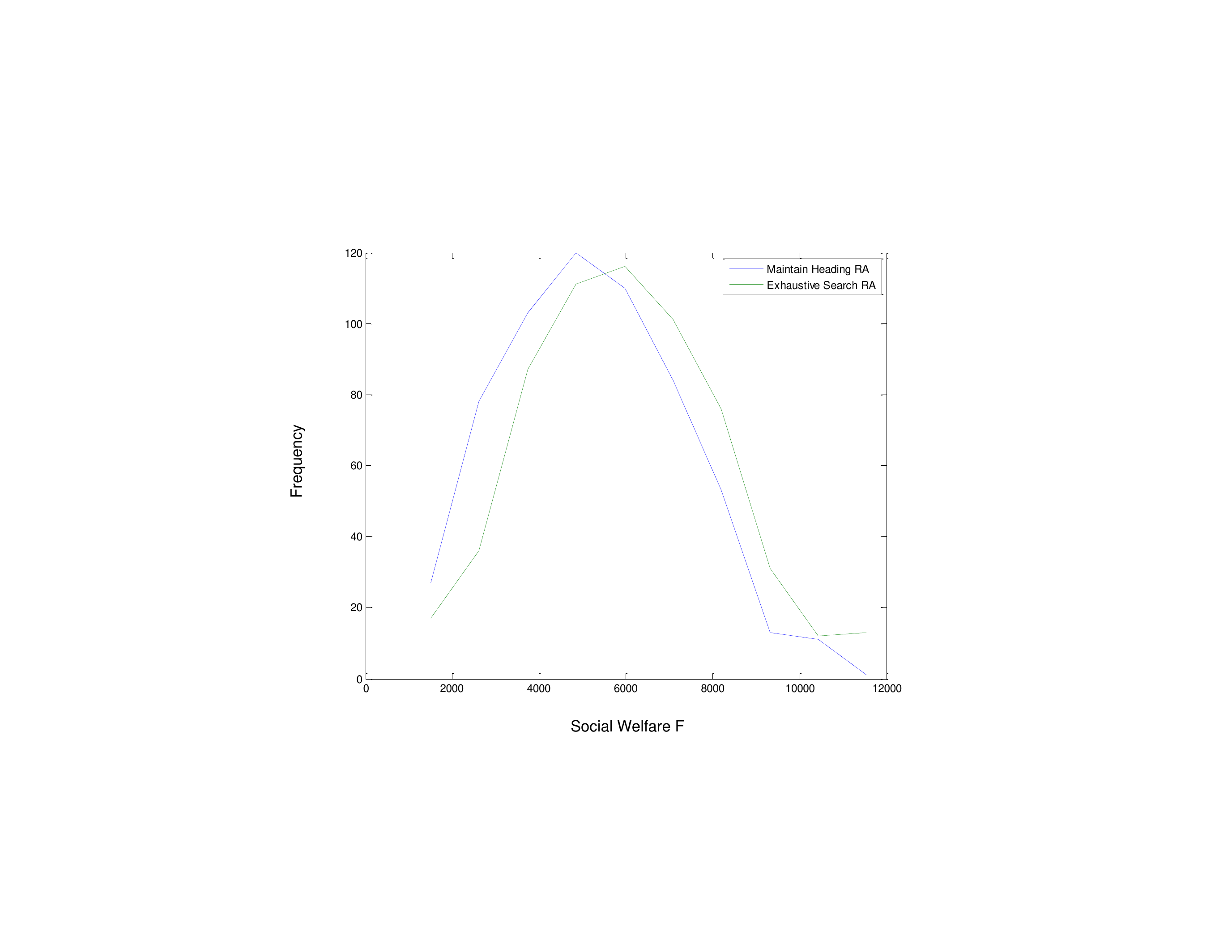}
\caption{A comparison of social welfare for two different horizontal \ac{ra} systems: Not only does the expected value of social welfare increase by using the exhaustive search method, we also observe a shift upwards in the entire probability distribution.}
\label{fig:horizG}       
\end{center}
\end{figure}

\section{Advantages of Semi Net-Form Game Modeling}
\label{sec:benefits}
There are many distinct benefits to using semi net-form game modeling.  We elaborate in the following section.

\begin{enumerate}
\item \textbf{Fully probabilistic.} Semi net-form game is a thoroughly probabilistic model, and thus represents all quantities in the system using random variables. As a result, not only are the probability distributions available for states of the Bayes net, they are also available for any metrics derived from those states.  For the system designer, the probabilities offer an additional dimension of insight for design.  For regulatory bodies, the notion of considering full probability distributions to set regulations represents a paradigm shift from the current mode of aviation operation.
\item \textbf{Modularity.} A huge strength to using a Bayes net as the basis of modeling is that it decomposes a large joint probability into smaller ones using conditional independence.  In particular, these smaller pieces have well-defined inputs and outputs, making them easily upgraded or replaced without affecting the entire net.  One can imagine an ongoing modeling process that starts by using very crude models at the beginning, then progressively refining each component into higher fidelity ones.  The interaction between components is facilitated by using the common language of probability.
\item \textbf{Computational human behavior model.} \ac{hitl} experiments (those that involve human pilots in mid- to high-fidelity simulation environments) are very tedious and expensive to perform because they involve carefully crafted test scenarios and human participants.  For the above reasons, \ac{hitl} experiments produce very few data points relative to the number needed for statistical significance.  On the other hand, semi net-form games rely on mathematical modeling and numerical computations, and thus produce data at much lower cost.
\item \textbf{Computational convenience.} Because semi net-form game algorithms are based on sampling, they enjoy many inherent advantages.  First, Monte Carlo algorithms are easily parallelized to multiple processors, making them highly scalable and powerful.  Secondly, we can improve the performance of our algorithms by using more sophisticated Monte Carlo techniques.
\end{enumerate}

\section{Conclusions and Future Work}
\label{sec:conclusion}
In this chapter, we defined a framework called ``Semi Network-Form Games," and showed how to apply that framework to predict pilot behavior in \acp{nmac}.  As we have seen, such a predictive model of human behavior enables not only powerful analyses but also design optimization.  Moreover, that method has many desirable features which include modularity, fully-probabilistic modeling capabilities, and computational convenience.

The authors caution that since this study was performed using simplified models as well as uncalibrated parameters, that further studies be pursued to verify these findings.  The authors point out that the primary focus of this work is to demonstrate the modeling technology, and thus a follow-on study is recommended to refine the model using experimental data.

In future work, we plan to further develop the ideas in semi network-form games in the following ways.  First, we will explore the use of alternative numerical techniques for calculating the conditional distribution describing a player's strategy $P(X_v \mid x_{pa(v)})$, such as using variational calculations and belief propagation~\protect{\cite{KollerBook}}.  Secondly, we wish to investigate how to learn semi net-form game model parameters from real data.  Lastly, we will develop a software library to facilitate semi net-form game modeling, analysis and design.  The goal is to create a comprehensive tool that not only enables easy representation of any hybrid system using a semi net-form game, but also houses ready-to-use algorithms for performing learning, analysis and design on those representations.  We hope that such a tool would be useful in augmenting the current verification and validation process of hybrid systems in aviation.  

By building powerful models such as semi net-form games, we hope to augment the  current qualitative methods (i.e., expert opinion, expensive \ac{hitl} experiments) with computational human models to improve safety and performance for all hybrid systems.

\begin{acknowledgement}
We give warm thanks to Mykel Kochenderfer, Juan Alonso, Brendan Tracey, James Bono, and Corey Ippolito for their valuable feedback and support.  We also thank the NASA \ac{irac} project for funding this work.
\end{acknowledgement}


\begin{thebibliography}{99.}%


\bibitem{BishopBook} Bishop, C.M.: Pattern recognition and machine learning.  Springer (2006)

\bibitem{Camerer10} Brunner, C., Camerer, C.F., Goeree, J.K.: A correction and re-examination of 'stationary concepts for experimental 2x2 games'. American Economic Review, forthcoming (2010)

\bibitem{Camerer89} Camerer, C.F.: An experimental test of several generalized utility theories. Journal of Risk and Uncertainty. \textbf{2}(1), 61--104 (1989)

\bibitem{CamererBook} Camerer, C.F.: Behavioral game theory: experiments in strategic interaction. Princeton University Press (2003)

\bibitem{Caplin10} Caplin, A., Dean, M., Martin, D.: Search and satisficing. NYU working paper (2010)

\bibitem{CostaGomes09} Costa-Gomes, M.A., Crawford, V.P., Iriberri, N.: Comparing models of strategic thinking in Van Huyck, Battalio, and Beil's coordination games. Journal of the European Economic Association (2009)

\bibitem{Crawford02} Crawford, V.P.: Introduction to experimental game theory (Symposium issue). Journal of Economic Theory. \textbf{104}, 1--15 (2002)

\bibitem{Crawford07} Crawford, V.P.: Level-k thinking. Plenary lecture. 2007 North American Meeting of the Economic Science Association. Tucson, Arizona (2007)

\bibitem{Crawford08} Crawford, V.P.: Modeling behavior in novel strategic situations via level-k thinking. GAMES 2008. Third World Congress of the Game Theory Society (2008)

\bibitem{Darwiche09} Darwiche, A.: Modeling and reasoning with Bayesian networks. Cambridge U. Press (2009)

\bibitem{Endsley88} Endsley, M.R.: Situation Awareness Global Assessment Technique (SAGAT). Proceedings of the National Aerospace and Electronics Conference (NAECON), pp. 789–-795. IEEE, New York (1988)

\bibitem{Endsley89} Endsley, M.R.: Final report: Situation awareness in an advanced strategic mission (No. NOR DOC 89-32). Northrop Corporation. Hawthorne, CA (1989)

\bibitem{FAA10} Federal Aviation Administration Press Release: Forecast links aviation activity and national economic growth (2010)
\\ 
\url{http://www.faa.gov/news/press_releases/news_story.cfm?newsId=11232} Cited 15 Mar 2011

\bibitem{KochenderferATC344} Kochenderfer, M.J., Espindle, L.P., Kuchar, J.K., Griffith, J.D.: Correlated encounter model for cooperative aircraft in the national airspace system. Massachusetts Institute of Technology, Lincoln Laboratory, Project Report ATC-344 (2008)

\bibitem{KochenderferATC345} Kochenderfer, M.J., Espindle, L.P., Kuchar J.K., Griffith, J.D.: Uncorrelated encounter model of the national airspace system. Massachusetts Institute of Technology, Lincoln Laboratory, Project Report ATC-345 (2008)

\bibitem{KochenderferATC360} Kochenderfer, M.J., Chryssanthacopoulos, J.P., Kaelbling, L.P., Lozano-Perez, T., Kuchar, J.K.: Model-based optimization of airborne collision avoidance logic. Massachusetts Institute of Technology, Lincoln Laboratory. Project Report ATC-360 (2010)

\bibitem{KochenderferItaly11} Kochenderfer, M.J., Chryssanthacopoulos, J.P.: Partially-cntrolled Markov decision processes for collision avoidance systems. International Conference on Agents and Artificial Intelligence. Rome, Italy (2011)

\bibitem{KochenderferATC371} Kochenderfer, M.J., Chryssanthacopoulos, J.P.: Robust airborne collision avoidance through dynamic programming.  Massachusetts Institute of Technology Lincoln Laboratory. Project Report ATC-371 (2011)

\bibitem{KollerBook} Koller, D., Friedman, N.: Probabilistic graphical models: principles and techniques. MIT Press (2009)

\bibitem{Koller03} Koller, D., Milch, B.: Multi-agent influence diagrams for representing and solving games. Games and Economic Behavior. \textbf{45}(1), 181--221 (2003)

\bibitem{Kuchar07} Kuchar, J.K., Drumm, A.C.: The Traffic Alert and Collision Avoidance System, Lincoln Laboratory Journal. \textbf{16}(2) (2007)

\bibitem{MyersonBook} Myerson, R.B.: Game theory: Analysis of conflict. Harvard University Press (1997)

\bibitem{Pearl00} Pearl, J.: Causality: Models, reasoning and inference. Games and Economic Behavior. Cambridge University Press (2000)

\bibitem{SfeArticle} Reisman, W.: Near-collision at SFO prompts safety summit. The San Francisco Examiner (2010)
\\ 
\url{http://www.sfexaminer.com/local/near-collision-sfo-prompts...-safety-summit} Cited 15 Mar 2011.

\bibitem{RobertBook} Robert, C.P., Casella, G.: Monte Carlo statistical methods 2nd ed. Springer (2004)

\bibitem{RussellBook} Russell, S., Norvig, P.: Artificial intelligence a modern approach. 2nd Edition. Pearson Education (2003)

\bibitem{Selten08} Selten, R., Chmura, T.: Stationary concepts for experimental 2x2 games. American Economic Review. \textbf{98}(3), 938--966 (2008)

\bibitem{Simon56} Simon, H.A.: Rational choice and the structure of the environment. Psychological Review. \textbf{63}(2), 129--138 (1956)

\bibitem{Simon82} Simon, H.A.: Models of bounded rationality. MIT Press (1982)

\bibitem{Taylor90} Taylor, R.M.: Situational Awareness Rating Technique (SART): The development of a tool for aircrew systems design. AGARD, Situational Awareness in Aerospace Operations 17 (SEE N90-28972 23-53) (1990)

\bibitem{Watts04} Watts, B.D.: Situation awareness in air-to-air combat and friction. Chapter 9 in Clausewitzian Friction and Future War, McNair Paper no. 68 revised ed. Institute of National Strategic Studies, National Defense University (2004)

\bibitem{WolpertNFG} Wolpert, D., Lee, R.: Network-form games: using Bayesian networks to represent non-cooperative games. NASA Ames Research Center working paper. Moffett Field, California (in preparation)

\bibitem{Wright10} Wright, J.R., Leyton-Brown, K.: Beyond equilibrium: predicting human behavior in normal form games. Twenty-Fourth Conference on Artificial Intelligence (AAAI-10) (2010)

%
\end{thebibliography}
\end{document}